\documentclass[runningheads]{llncs}
\usepackage[utf8]{inputenc}
\usepackage[english]{babel}
\usepackage{graphicx}
\usepackage{float}
\usepackage{mathtools}
\usepackage{xcolor}
\usepackage{xspace}
\usepackage{bbding}
\usepackage{algorithm}
\usepackage[noend]{algpseudocode}
\usepackage{caption}
\usepackage{subcaption}
\usepackage{marginnote}
\usepackage{wrapfig}
\usepackage{amssymb,amsmath}
\usepackage{enumerate}

\usepackage{paralist}
\usepackage{xspace}

\newcommand{\sayan}[1]{\textcolor{blue}{#1}}
\newcommand{\kristina}[1]{\textcolor{red}{#1}}

\definecolor{asparagus}{rgb}{0.53, 0.66, 0.42}

\newcommand{\rmv}[1] {}

\newcommand{\num}[1]{\relax\ifmmode \mathbb #1\else $\mathbb #1$\fi}
\newcommand{\nnnum}[1]{\relax\ifmmode
  {\mathbb #1}_{\geq 0} \else ${\mathbb #1}_{\geq 0}$
  \fi}
\newcommand{\npnum}[1]{\relax\ifmmode
  {\mathbb #1}_{\leq 0} \else ${\mathbb #1}_{\leq 0}$
  \fi}
\newcommand{\pnum}[1]{\relax\ifmmode
  {\mathbb #1}_{> 0} \else ${\mathbb #1}_{> 0}$
  \fi}
\newcommand{\nnum}[1]{\relax\ifmmode
  {\mathbb #1}_{< 0} \else ${\mathbb #1}_{< 0}$
  \fi}
\newcommand{\plnum}[1]{\relax\ifmmode
  {\mathbb #1}_{+} \else ${\mathbb #1}_{+}$
  \fi}
\newcommand{\nenum}[1]{\relax\ifmmode
  {\mathbb #1}_{-} \else ${\mathbb #1}_{-}$
  \fi}

\newcommand{\reals}{{\num R}}                    %reals
\newcommand{\A}{\mathcal{A}}
\newcommand{\B}{\mathcal{B}}
\newcommand{\D}{\mathcal{D}}
\newcommand{\M}{\mathcal{M}}
\renewcommand{\P}{\mathcal{P}}
\newcommand{\R}{\mathcal{R}}
\newcommand{\SL}{\mathcal{\pi}}
\newcommand{\U}{\mathcal{U}}

\newcommand{\SC}{{\mathtt S}}
\newcommand{\UC}{{\mathtt U}}

\newcommand{\Unsafe}{\mathcal{B}}

\newcommand{\supp}[1]{{\rm{supp}\mathit{(#1)}}}

\newcommand{\tuple}[1] {\langle #1 \rangle}
\newcommand{\plant} {plant}

\newcommand{\policy} {\mathsf{SP}}
\newcommand{\stpolicy} {\mathsf{StationarySP}}
\newcommand{\sfpolicy} {\mathsf{SafeSP}}
\newcommand{\unsafe} {\mathsf{Unsafe}}
\newcommand{\prob} {\mathbb{P}}
\newcommand{\runs} {\mathsf{Runs}}
\newcommand{\finruns} {\runs_f}
\newcommand{\prefix}[2][i]{#2[:#1]}
\newcommand{\cyl}[1] {C_{#1}}
\newcommand{\set}[1] {\{#1\}}
\newcommand{\setpred}[2] {\set{#1\: |\: #2}}

\newcommand{\lead}{\ell}
\newcommand{\ego}{}

\newcommand{\reachrta}{\textsf{ReachRTA}}
\newcommand{\rlrta}{\textsf{RLRTA}}

\newcommand{\simrta}{\textsf{SimRTA}}
\newcommand{\lstate}{\mathit{lstate}}
\newcommand{\reach}{\mathit{Reach}}

\newcommand{\acc}{\textsf{Acc}}
\newcommand{\dubins}{\textsf{Dubins}}
\newcommand{\air}{\textsf{Aircraft}}

\newcommand{\pitch}{\gamma}
\newcommand{\yaw}{\psi}
\newcommand{\pitchInput}{\Gamma}
\newcommand{\yawInput}{\omega}

\newcommand{\building}{\textsf{Dubins+O}}
\newcommand{\groundCollision}{\textsf{Air}}
\newcommand{\rf}{\textrm{ref}}
\newcommand{\err}{\varepsilon}

\usepackage{listings}

\usepackage{stmaryrd}
\usepackage{multirow}
\usepackage{array}
\usepackage{upgreek}
\usepackage{makecell}

\usepackage{fancyhdr}

\fancypagestyle{Distribution}{
    \fancyhf{} %Clear header/footer
    \fancyfoot[C]{DISTRIBUTION STATEMENT A. Approved for public release: distribution is unlimited.} %Centered footer
    }

\lstdefinelanguage{pseudocode}{
	basicstyle=\scriptsize,
	keywordstyle=\bf \scriptsize,
	identifierstyle=\it \scriptsize,
	mathescape=true,
	tabsize=20,
	xleftmargin=4.0ex,
	sensitive=false,
	columns=fullflexible,
	keepspaces=false,
	basewidth=0.05em,
	moredelim=[il][\rm]{//},
	moredelim=[is][\sf \figuresize]{!}{!},
	moredelim=[is][\bf \figuresize]{*}{*},
	keywords={automaton, algorithm, and,
		break,
		choose,const,continue, components,
		discrete, do,
		eff, external,else, elseif, evolve, end, each, exit,
		fi,for, forward, from, find,
		hidden,
		in,input,internal,if,invariant, initially, imports,
		let,
		mode,
		or, output, operators, od, of,
		pre,
		return,
		such,satisfies, stop, signature, simulation, sample,
		trajectories,trajdef, transitions, that,then, type, types, to, tasks,
		variables, vocabulary,
		when,where, with,while},
	emph={set, seq, tuple, map, array, enumeration},
	literate=
	{(}{{$($}}1
	{)}{{$)$}}1
	{\\in}{{$\in\ $}}1
	{\\preceq}{{$\preceq\ $}}1
	{\\subset}{{$\subset\ $}}1
	{\\subseteq}{{$\subseteq\ $}}1
	{\\supset}{{$\supset\ $}}1
	{\\supseteq}{{$\supseteq\ $}}1
	{\\forall}{{$\forall$}}1
	{\\le}{{$\le\ $}}1
	{\\ge}{{$\ge\ $}}1
	{\\gets}{{$\gets\ $}}1
	{\\cup}{{$\cup\ $}}1
	{\\cap}{{$\cap\ $}}1
	{\\langle}{{$\langle$}}1
	{\\rangle}{{$\rangle$}}1
	{\\exists}{{$\exists\ $}}1
	{\\bot}{{$\bot$}}1
	{\\rip}{{$\rip$}}1
	{\\emptyset}{{$\emptyset$}}1
	{\\notin}{{$\notin\ $}}1
	{\\not\\exists}{{$\not\exists\ $}}1
	{\\ne}{{$\ne\ $}}1
	{\\to}{{$\to\ $}}1
	{\\implies}{{$\implies\ $}}1
	{<}{{$<\ $}}1
	{>}{{$>\ $}}1
	{=}{{$=\ $}}1
	{~}{{$\neg\ $}}1
	{|}{{$\mid$}}1
	{'}{{$^\prime$}}1
	{\\A}{{$\forall\ $}}1
	{\\E}{{$\exists\ $}}1
	{\\/}{{$\vee\,$}}1
	{\\vee}{{$\vee\,$}}1
	{/\\}{{$\wedge\,$}}1
	{\\wedge}{{$\wedge\,$}}1
	{=>}{{$\Rightarrow\ $}}1
	{->}{{$\rightarrow\ $}}1
	{<=}{{$\Leftarrow\ $}}1
	{<-}{{$\leftarrow\ $}}1
	{~=}{{$\neq\ $}}1
	{\\U}{{$\cup\ $}}1
	{\\I}{{$\cap\ $}}1
	{|-}{{$\vdash\ $}}1
	{-|}{{$\dashv\ $}}1
	{<<}{{$\ll\ $}}2
	{>>}{{$\gg\ $}}2
	{||}{{$\|$}}1
	{[}{{$[$}}1
	{]}{{$\,]$}}1
	{[[}{{$\langle$}}1
	{]]]}{{$]\rangle$}}1
	{]]}{{$\rangle$}}1
	{<=>}{{$\Leftrightarrow\ $}}2
	{<->}{{$\leftrightarrow\ $}}2
	{(+)}{{$\oplus\ $}}1
	{(-)}{{$\ominus\ $}}1
	{_i}{{$_{i}$}}1
	{_j}{{$_{j}$}}1
	{_{i,j}}{{$_{i,j}$}}3
	{_{j,i}}{{$_{j,i}$}}3
	{_0}{{$_0$}}1
	{_1}{{$_1$}}1
	{_2}{{$_2$}}1
	{_n}{{$_n$}}1
	{_p}{{$_p$}}1
	{_k}{{$_n$}}1
	{-}{{$\ms{-}$}}1
	{@}{{}}0
	{\\delta}{{$\delta$}}1
	{\\R}{{$\R$}}1
	{\\Rplus}{{$\Rplus$}}1
	{\\N}{{$\N$}}1
	{\\times}{{$\times\ $}}1
	{\\tau}{{$\tau$}}1
	{\\alpha}{{$\alpha$}}1
	{\\beta}{{$\beta$}}1
	{\\gamma}{{$\gamma$}}1
	{\\ell}{{$\ell\ $}}1
	{\\TT}{{\hspace{1.5em}}}3
}

\lstdefinelanguage{pseudocodeNums}[]{pseudocode}
{
	numbers=left,
	numberstyle=\tiny,
	stepnumber=2,
	numbersep=4pt
}

\usepackage{hyperref}
\hypersetup{
	pdftitle={Optimistic Optimization for Statistical Model Checking with Regret Bounds},
	pdfauthor={Sayan Mitra},
	colorlinks=true,
	citecolor={blue},
	linkcolor = {blue},
	pagecolor={blue},
	backref={true},
	bookmarks=true,
	bookmarksopen=false,
	bookmarksnumbered=true
}

\begin{document}

% \usepackage{fancyhdr}

% \pagestyle{fancy}
% %... then configure it.
% \fancyhead{} % clear all header fields
% \fancyhead[RO,LE]{\textbf{The performance of new graduates}}
% \fancyfoot{} % clear all footer fields
% \fancyfoot[LE,RO]{\thepage}
% \fancyfoot[LO,CE]{From: K. Grant}
% \fancyfoot[CO,RE]{To: Dean A. Smith}

\pagestyle{fancy}
\title{Searching for Optimal Runtime Assurance via Reachability and Reinforcement Learning\thanks{Miller and Zeitler made equal contributions and the other authors are listed alphabetically.}}
\author{Kristina Miller\inst{1,*} 
%University of Illinois at Urbana Champaign) <kmmille2@illinois.edu>
\and 
Christopher K. Zeitler\inst{2,*} 
%(Rational CyPhy Inc.) <ckzeitler@gmail.com>
\and
William Shen\inst{1} 
%(University of Illinois at Urbana Champaign) <wshen15@illinois.edu>
\and
Kerianne Hobbs\inst{3}  
%(Air Force Research Laboratory) <kerianne.hobbs@afrl.af.mil>
\and
Sayan Mitra\inst{1} 
%(University of Illinois at Urbana Champaign) <mitras@illinois.edu>
\and
John Schierman\inst{3}  
%(Air Force Research Laboratory) <john.schierman.1@afrl.af.mil>
\and 
Mahesh Viswanathan\inst{1} 
%(University of Illinois at Urbana Champaign) <vmahesh@illinois.edu>}
% KD?
\institute{$^1$University of Illinois at Urbana-Champaign \\
$^2$Rational CyPhy Inc. \\
$^3$Air Force Research Laboratory}
\email{\small \{kmmille2,wshen15,mitras,vmahesh\}@illinois.edu} \\
\email{\small ckzeitler@gmail.com} \\
\email{\small \{john.schierman.1,kerianne.hobbs\}@us.af.mil}
}
\maketitle
\thispagestyle{fancy}
\begin{abstract}
A runtime assurance system (RTA) for a given plant enables the exercise of an untrusted or experimental controller  while assuring safety with a backup (or safety) controller. The relevant computational design problem is to create a logic that assures safety by switching to the safety controller as needed, while maximizing some performance criteria, such as the utilization of the untrusted controller. Existing RTA design strategies are well-known to be overly conservative and, in principle, can lead to safety violations.  In this paper, 
we formulate the optimal RTA design problem and present a new approach for solving it. Our approach relies on reward shaping and reinforcement learning. It can guarantee safety and leverage machine learning technologies for scalability. We have implemented this algorithm and present experimental results comparing our approach with state-of-the-art reachability and simulation-based RTA approaches in a number of scenarios using aircraft models in 3D space with complex safety requirements. Our approach can guarantee safety while increasing utilization of the experimental controller over existing approaches. 
\end{abstract}

\section{Introduction}
\label{sec:intro}

A well-designed {\em Runtime Assurance System (RTA)\/} ensures the safety of a control system while  switching between an untrusted controller and a safety controller to maximimize some performance criteria. The  motivation for studying the RTA design problem is straightforward (see the system architecture shown in Figure~\ref{fig:arch}):  The plant is a state machine $\M$ with a specified set of unsafe states ($\B$). The {\em safety\/} ($\SC$) and the {\em untrusted controllers\/} ($\UC$)  determine the transitions in $\M$. Roughly, the  optimal RTA problem is to design a {\em switching policy\/} $\pi$ that chooses the controller at each step with the twin goals of (a) keeping the closed-loop system ($\M$ with  $\pi$) safe {\em and\/} (b) optimizing some other performance objective, such as maximally using the untrusted controller.

RTA is seen as a key enabler for autonomy in safety-critical aerospace systems~\cite{RTA-AFC2010,RTA-AAS2020,Hobbs2021RunTA,RTA-JAX-abs-2209-01120}. 
RTA systems have  enabled sandboxing and fielding of experimental technologies for 
spacecraft docking~\cite{DunlapHMH22},
learning-enabled taxiing~\cite{cofer-taxi-20}, air-collision avoidance~\cite{cofer-rta-22}, flight-critical control~\cite{RTA-AFC2010},  autonomous aero-space systems~\cite{RTA-AAS2020}, and multi-agent formation flight for UAVs~\cite{RTA-Distt2013-BakHuang,sandbox11}, all within harsh environments, while maintaining safety and verifiability.
 %
% RTA have been deployed in flight tests for protection against software faults in flight-critical systems of unmanned aerial vehicles (UAV)~\cite{RTA-AFC2010}.
 %
 More than 200 papers have been written on this topic since 2020.

A natural  switching policy is to use lookahead or forward simulations~\cite{bak2009system,wadley2013development,hobbs2018space}. From a given state $q$, if the simulations of the closed loop system with the untrusted controller ($\M$ + $\UC$), up to a time horizon $T>0$, are safe, then $\UC$ is allowed to continue; otherwise, the policy switches to $\SC$. This strategy is employed, for example,  in Wadley et al.~\cite{wadley2013development} for an air collision avoidance system. To cover for the uncertainties arising from the  measurement of the initial state $q$ and the model $\M$ of the plant, forward simulation can be replaced with computation of the {\em reachable set\/}, $\reach^\UC(\hat{Q},T)$,  of $\M + \UC$ up to time horizon $T$ from the uncertain initial set $\hat{Q}$. If the computed reachable set $\reach^\UC(\hat{Q},T)$ does not intersect with $\B$ (unsafe set), then the unverified controller is allowed to continue; otherwise, there is a switch to $\SC$. 
In some studies~\cite{hobbs2018space,sandbox11}, it is shown that this is a viable strategy for collision detection and safe operation of systems.

\begin{figure}[h!]
	\centering
	\includegraphics[width=0.5\textwidth]{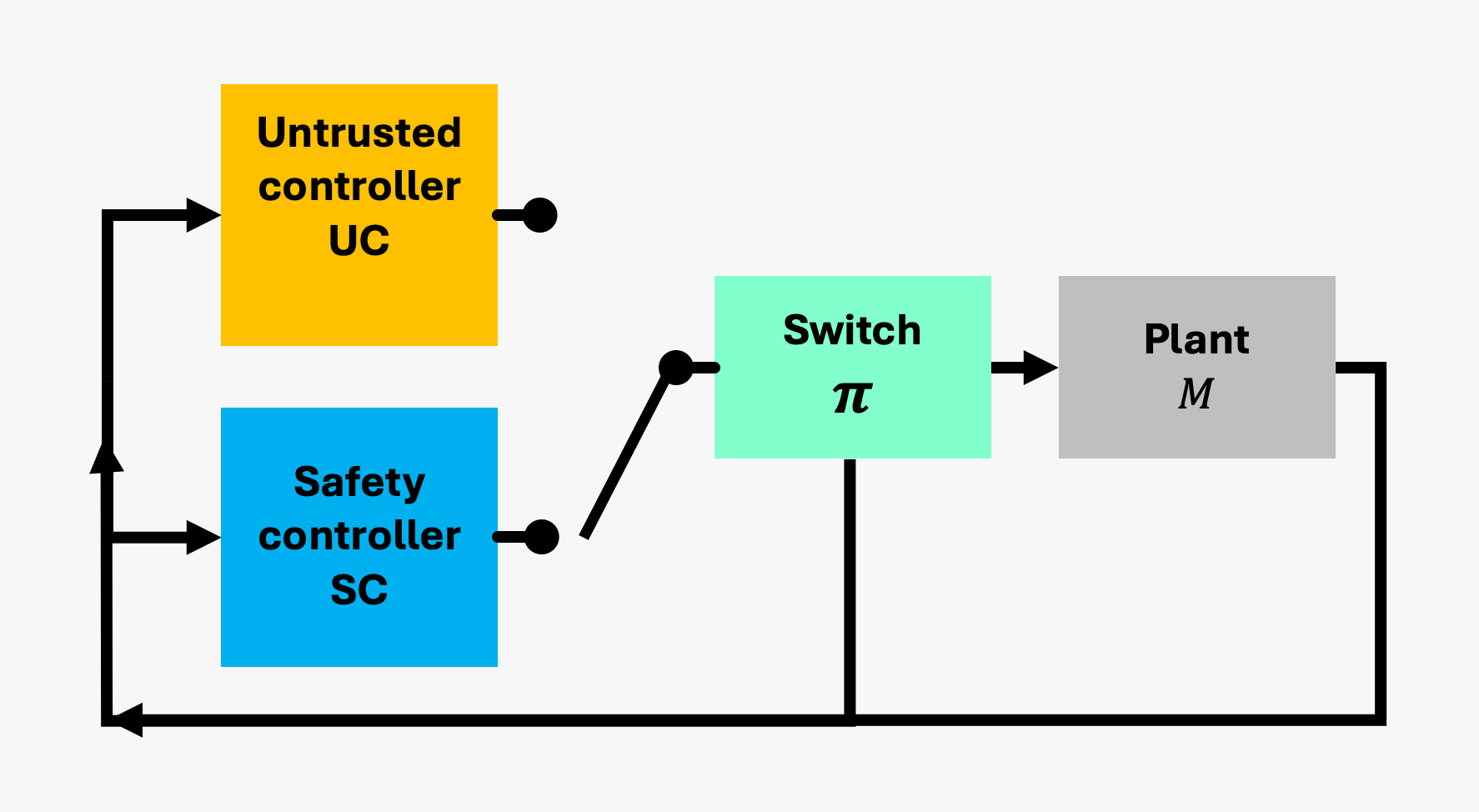}
	\caption{\small Runtime Assurance (RTA) system architecture.}
	\label{fig:arch}
\end{figure}
 
It has been known that these popular policies actually do not guarantee safety of the closed loop system. We illustrate this with counter examples in Section~\ref{sec:basicRS}. Because of inertia and delays, the safety controller may be doomed to violate safety  by the time the switching policy hands off control to it. It has also been known that, instead of switching when $\reach^\UC(\hat{Q},T)$ exits the safe set $\B^c$, the switch should occur when it exits the {\em recoverable set\/}, an invariant subset of safe states from which the safety controller $\SC$ is guaranteed to preserve safety. We show with another counter example (Example~\ref{ex:recoverable-suboptimal}) that using the recoverable set-based switching is safe but not optimal in the sense that it may rule out   perfectly safe opportunities for using $\UC$.

Observing that the popular strategies for designing RTA strategies may not be safe or optimal, we put forward an alternative perspective on the RTA design problem.
The requirement of maximally utilizing the untrusted controller can be specified in different ways---for example, using counters or temporal logics. More generally, we can see this as a problem of maximizing certain rewards defined on the transitions of $\M$ while assuring hard requirements like safety. 
This leads to a generalized formulation of the  RTA design problem: Given a plant automaton  $\M$ with a set of actions $A$ (corresponding to the choice of different controllers, possibly including but not restricted to $\SC$ and $\UC$), a reward structure $r$, and  an unsafe set $\B$,  the goal is to find a (possibly randomized and history dependent) switching policy $\pi$, so that the resulting executions are always safe and maximize the expected reward. 

We show that this constrained optimization problem for $(\M, \B)$ $\textendash$ maximizing expected reward while guaranteeing safety $\textendash$ has an optimal controller switching policy that is \emph{deterministic} and depends only on the current state. Such policies are called \emph{stationary}. The existence of optimal policies that are also stationary is not a given, especially in a context when multiple objectives are being considered. For example, the problem of finding an optimal switching policy that guarantees a target state is reached may not have \emph{any} optimal policy, let alone a stationary one (see example and discussion in Section~\ref{sec:conc}). We then consider the problem of algorithmically designing an optimal policy that guarantees safety. This problem has received much attention recently, and there have been proposals to adapt reinforcement learning to solve it (see related work discussion in Section~\ref{sec:related}). In this paper, we take a slightly different approach by reducing the RTA design problem to that of finding a switching policy that optimizes a \emph{single} objective function, which encodes the requirement of being safe  within a modified reward function. Thus, our reward shaping method allows one to use off-the-shelf, vanilla reinforcement learning methods to solve the RTA problem.

In summary, the contributions of this work are as follows: 
\begin{itemize}
    % \item We formalize  folklore  about the RTA design problem, namely, that safe set lookahead strategies can be unsafe and that recoverable set lookahead strategies can be suboptimal. 
    \item We address the generalized RTA design problem of searching for safe optimal policies over the family of randomized and history-dependent switching policies for automata with rewards. 
    \item We show that there exists a reward shaping strategy such that the  modified, reward-shaped, automaton $\M'$ will have a  stationary optimal policy. And,  if the  optimal memoryless policy $\pi$ for $\M'$ has a non-negative reward, then $\pi$ is also a safe optimal policy for the original automaton $\M$. 
    \item We  implement the proposed reward-shaping\textendash based RTA design procedure using standard reinforcement learning libraries, and we empirically compare this $\rlrta$ against the existing approaches. $\rlrta$ can find safe switching policies in complex scenarios involving leader-follower vehicles and obstacles, in one, two, and three dimensions. In most scenarios, $\rlrta$ improves the utilization of the untrusted controller significantly.
\end{itemize}

\subsection{Related Work}
\label{sec:related}

The idea of RTA is related to myriad other concepts such as the Simplex method and its variations~\cite{sha2001using}, supervisory control~\cite{ramadge1987supervisory}, runtime monitoring, and shielding~\cite{alshiekh2018safe}. The computational challenge of implementing all of these ideas lies in solving the problem we address.
These papers~\cite{rta-cps-all,RTA-VV-SC-FCS-2008,Hobbs2021RunTA} provide an excellent overview of the state of the art and its challenges. 
%he issue of conservatism in the RTA logic is identified and addressed in~\cite{RTA-conservatism-2012} in the context of an Min-Max control selection structure. 

\begin{table}[!ht]
    \centering
    \small
    \begin{tabular}{|l|l|c|c|c|}
    \hline
        Method  & Decision Technique & Safe   & Optimal  & Model-free  \\ \hline
        Black-box Simplex~\cite{MehmoodSBSS22} & Reachability & $\checkmark^*$  & $\mathsf{x}$ & $\mathsf{x}$ \\ \hline
        Sandboxing~\cite{bak2011sandboxing} & Reachability & $\mathsf{x}$ & $\mathsf{x}$ & $\mathsf{x}$ \\ \hline
        SOTER~\cite{SOTER} & Reachability & $\checkmark$ & $\mathsf{x}$ & $\mathsf{x}$\\ \hline
        Active Set Invariance Filtering~\cite{ames2019control,gurriet2018online} & Control barrier funcs & $\checkmark$ & $\mathsf{x}$ & $\mathsf{x}$ \\ \hline
        Optimal RTA  (our work)  & Reward shaping, RL & $\checkmark$ & $\checkmark$ & $\checkmark$ \\ \hline
        Shielding~\cite{alshiekh2018safe,konighofer2020shield} & Game solving & $\checkmark$ & $\checkmark$ & $\mathsf{x}$ \\ \hline
    \end{tabular}
    \vspace{0.25cm}
    \caption{\small A classification of solutions of the RTA design problem. {\em Safe\/} means that the method is guaranteed to assure unbounded time safety of the overall system. {\em Optimal} means that the method can  design switching policies that maximizes additional objectives (other than safety). {\em Model-free\/} means that the method does not require analytical models for the plant and the controllers. The $\checkmark^*$ annotation in the ``Safe'' column for the Black-box Simplex technique is to indicate that safety is not guaranteed but can be ensured if additionally all the safe states form a recoverable set for the safety controller.}
\end{table}

Reachability analysis has been used for RTA extensively~\cite{BakHAC15,bak2011sandboxing,musau2022icaa}. The  idea is to compute the set of  reachable states from the current state (typically using system models) and switch to the safety controller when this reachable set violates safety. We compare our approach against this strategy, which we call  $\reachrta$  in Section~\ref{sec:exp}. Instead of computing this switching logic statically, some studies~\cite{musau2022icaa,SibaiOnlinePedest20} perform the reachability online. 
In Section~\ref{sec:basicRS}, we will see that this strategy is not safe; it can be made safe by switching at the {\em recoverable set\/} instead of the unsafe set. The recoverable set is typically computed by performing backward reachability from the unsafe states, which requires access to the model of the plant. 
In Mehmood et al.~\cite{MehmoodSBSS22},  the computation of the recoverable set is replaced, in a sense, with a precomputed permanent sequence of safe moves for the safety controller.

%\sayan{How does SOTER (https://arxiv.org/pdf/1808.07921.pdf)  relate to this narrative? Their dev of the inv $\phi$ in Theorem 3.1 seems to assume that the safety controller can always recover from a ``safer'' set?}

%\paragraph{Blending Controls. }
An alternative to designing a switching logic is to design a {\em filter\/} that blends the outputs of the  $\SC$ and the $\UC$ to create the final control output. The intuition is to move the $\UC$ output in the direction of $\SC$ when the former may violate safety. 
Active set invariance filtering (ASIF)~\cite{ames2019control,gurriet2018online} is an instance of this approach, and it uses control barrier functions (CBF), which can be constructed using the safety controller and the plant models.  These methods have been employed for satellites performing autonomous operations \cite{hibbard2022guaranteeing,mote2021natural,dunlap2022run}.
When the unsafe sets are not known a priori, constructing CBFs can be challenging. Further, ASIF aims to be minimally invasive but does not accommodate additional performance criteria that are optimized

% \sayan{Is this fair? What can be said about achieving optimality by this approach?} \mahesh{I doubt if this technique will guarantee optimality with respect to rewards because rewards don't seem to be used in the filtering.} \chris{We may want to make this explicit, because the ASIF paper dicusses 'optimality' in the sense that it is minimally invasive in regards to modifying the unverified control}
% \sayan{Where has it been used? }

\paragraph{RL to Verify Temporal Logic Specifications.}
There has been extensive work on using RL to solve the temporal logic verification problem for MDPs, where, given an MDP and a specification $\varphi$, the goal is to use RL to find a policy that maximizes (or minimizes) the probability of satisfying specification $\varphi$~\cite{sadigh2014learning,hasanbeig2019reinforcement,hahn2019omega,bozkurt2020control,lavaei20,alur_framework_2021,balakrishnan22}. The techniques used in this approach are similar: reward shaping, adding terminal reject states, etc. However, fundamentally, the problem being solved in these papers is different; in these papers there is only one objective being optimized, while in this paper we need to optimize rewards while satisfying a hard constraint like safety. The presence of two, possibly competing, objectives $\textendash$ maximizing reward while satisfying a hard safety requirement $\textendash$ changes the problem significantly. One examples of this is that there is always an optimal policy that maximizes the probability of satisfying a temporal property, but there is \emph{no} optimal policy that maximizes reward while satisfying a hard reachability constraint (see Section~\ref{sec:conc} for an example).

\paragraph{Safe RL.}
There is a lot of interest in developing RL-based approaches to ensure ``safety,'' which means different things in different contexts (see the survey articles for review~\cite{GarciaF15,SafeRL22}).
The approaches include modification of the optimization criteria, introduction of risk sensitivity, robust model estimation, and risk-directed policy exploration.
%\mahesh{Sayan, you had read a survey on this. Wondering if you could add a reference to the survey and a couple of sentences on the approaches you found at the beginning.}
The closest work in this space to ours is the use of \emph{shielding} to guarantee that a safe optimal policy for an MDP is found~\cite{alshiekh2018safe,konighofer2020shield,shieldingpomdp22}. This approach synthesizes a shield that identifies ``safe'' transitions from each state. A safe transition is one that takes a system to another state from which there is some policy that can ensure safety. Optimal policy is then synthesized using RL with the shield. Synthesizing such a shield involves solving a multiplayer game which requires access to the plant model and can be computationally expensive. In contrast, our approach of reward shaping needs no expensive pre-computation. Finally, if shielding is used post-training, then the resulting policy is not guaranteed to have the optimal reward. Further, the shield needs to always be in place during deployment.

\section{The RTA Problem}
\label{sec:reward}

In a Runtime Assurance (RTA) framework, the central goal is to design a switching policy that chooses between a safe and an untrusted controller at each step, ensuring that the system remains safe while some objective (potentially the use of the untrusted controller) is maximized. To model this setup formally and to define the computational problem, it is convenient to model the state space of the control system/plant whose safety the RTA is trying to ensure.

\paragraph{Plants.}
A \emph{\plant} is tuple $\M = \tuple{Q,q_0,A,\D}$ where $Q$ is the (finite) set of states, $q_0 \in Q$ is the initial state, $A$ is a finite set of actions, and $\D: Q \times A \to Q$ is the transition function. In this paper, we will assume that the {\plant} has a large, but finitely many states and its transitions are \emph{discrete}. Most control systems can be reasonably approximated in such a manner by quantizing the state space and time. The transition function identifies the state of the system after following the control law defined by an action for a discrete duration of time. Notice, we are assuming that every action in $A$ is \emph{enabled} from each state; this is for modeling convenience and is not a restriction. Finally, in the context of RTA, the action set $A$ has typically only two elements, namely, using the safe $\SC$ or the untrusted controller $\UC$. 

\paragraph{Switching Policies.}
A \emph{switching policy} $\pi$ determines the action to be taken at each step. In general, a switching policy's choice depends on the computation history until that step (i.e., the sequence of states visited and actions taken), and the result of the roll of a dice. Thus, mathematically, it can be defined as follows. For any finite set $B$, $\prob(B)$ is the set of probability distributions over $B$. A \emph{run/finite run} of {\plant} $\M  = \tuple{Q,q_0,A,\D}$ is an alternating infinite/finite sequence of states and actions $\tau = p_0,a_0,p_1,a_1,\ldots p_n, \ldots$, where $p_i \in Q$, $a_i \in A$, $p_0 = q_0$, and $\D(p_i,a_i) = p_{i+1}$; a finite run is assumed to end in a state. The set of all runs (finite runs) of $\M$ will be denoted as $\runs(\M)$ ($\finruns(\M)$). A switching policy for {\plant} $\M$ is a function $\pi: \finruns(\M) \to \prob(A)$. The set of switching policies of a {\plant} $\M$ will be denoted by $\policy(\M)$. We will often be interested in a special class of switching policies called \emph{stationary} switching policies. A stationary policy is one where the choice of the next action is \emph{memoryless}, i.e., depends only on the last state of the run and not the entire run, and is \emph{deterministic}, i.e., exactly one action has non-zero probability. Such stationary policies can be represented as a function $\pi: Q \to A$. The set of all stationary policies of $\M$ will be denoted as $\stpolicy(\M)$.

\paragraph{Probability of Runs.}
A switching policy $\pi$ assigns a probability measure to runs in a standard manner. For a finite run $\tau = p_0,a_0,p_1,\ldots p_n \in \finruns(\M)$, \emph{the cylinder set} $\cyl{\tau}$ is the set of all runs $\rho$ that have $\tau$ as a prefix. The probability of $\cyl{\tau}$ under $\pi$ is given by 
\[
P_\pi(\cyl{\tau}) = \prod_{i=0}^{n-1} \pi(\prefix{\tau})(a_i),
\]
where $\prefix{\tau}$ is the $i$ length prefix $q_0,a_0,q_1,\ldots q_i$. The probability measure $P_\pi$ extends to a unique probability on the $\sigma$-field generated by the cylinder sets $\setpred{\cyl{\tau}}{\tau \in \finruns(\M)}$. If $\pi$ is stationary then $P_\pi(\rho) \neq 0$ for exactly one run $\rho$ and in that case $P_\pi(\rho) = 1$.

\paragraph{Safe Policies.}
The principal goal of a switching policy is to ensure the safety of the system. We will assume safety is modeled through an \emph{unsafe set} $\B$.  This unsafe set might represent different types of safety constraints, such as avoiding a region the system must not enter, such as in geofencing scenarios, or ensuring the system never depletes its fuel, and need not be limited to avoiding collisions.  Given a {\plant} $\M = \tuple{Q,q_0,A,\D}$ and an unsafe set $\B$, a run $\rho = p_0,a_0,p_1,a_1,\ldots$ is said to be \emph{unsafe} if there is an $i$ such that $p_i \in \B$. The set of unsafe runs of $\M$ with respect to $\B$ will be denoted as $\unsafe(\M,\B)$. The set $\unsafe(\M,\B)$ is measurable. A policy $\pi$ will be called \emph{safe} for $\M$ with respect to $\B$ if $P_\pi(\unsafe(\M,\B)) = 0$. The set of safe policies will be denoted as $\sfpolicy(\M,\B)$.

\paragraph{Rewards.}
A switching policy for an RTA is often designed to meet certain goals that include maximizing a reward or minimizing a cost. For example, the goal of an RTA could be to maximize the use of the untrusted controller, or maximize the time spent by the system in a target region. It may also be to use the safe and untrusted controllers in such a way as to minimize costs like fuel consumption. These can be formally captured in our setup through reward structures. A \emph{reward structure} on a {\plant} $\M = \tuple{Q,q_0,A,\D}$ is a pair $(r,\gamma)$, where $r: Q \times A \to \reals$ is the \emph{reward function}, and $\gamma \in (0,1)$ (open interval between $0$ and $1$) is the \emph{discount factor}. A reward structure assigns a reward to every run of $\M$ as follows. For $\rho = p_0,a_0,p_1,a_1,\ldots \in \runs(\M)$
\[
r(\rho) = \sum_{i} \gamma^i r(p_i,a_i).
\]
Because of the discount factor, the above infinite sum converges. The \emph{reward of policy} $\pi$ is the expected reward of runs based on the probability distribution $P_\pi$, i.e., $r(\pi) = E_{\rho \sim P_\pi}[r(\rho)]$. The calculation of a reward becomes simple in the case of a stationary policy; since exactly one run $\rho$ has non-zero probability under $\pi$, $r(\pi) = r(\rho)$, where $P_\pi(\rho) = 1$.

\paragraph{Optimal Policies.}
Let us fix a {\plant} $\M = \tuple{Q,q_0,A,\D}$, an unsafe set $\B$, and a reward structure $(r,\gamma)$. The goal in RTA is to find policies that maximize the reward. Let us define this precisely. The \emph{value} of a {\plant} $\M$ with respect to $(r,\gamma)$ is the supremum reward that can be achieved through a policy, i.e., 
$$
V(\M,r,\gamma) = \sup_{\pi \in \policy(\M)} r(\pi).
$$ 
A switching policy $\pi$ is \emph{optimum} for $\M$ and $(r,\gamma)$ if $r(\pi) = V(\M,r,\gamma)$. In general, optimum policies may not exist as the supremum may not be achieved by a policy. In RTA, we will often be interested in restricting our attention to safe policies. The \emph{value} of $\M$ with respect to $(r,\gamma)$ and $\B$ is given by 
\[
V(\M,r,\gamma,\B) = \sup_{\pi \in \sfpolicy(\M,\B)} r(\pi).
\]
As always, the supremum over an empty set is $-\infty$; thus, when $\sfpolicy(\M,\B) = \emptyset$, $V(\M,r,\gamma,\B) = -\infty$. Finally an optimal safe policy is $\pi \in \sfpolicy(\M,\B)$ such that $r(\pi) = V(\M,r,\gamma,\B)$. Again an optimal safe policy may not exist.

\paragraph{The RTA Problem.}
The goal of RTA is to find an optimal safe switching policy. Thus as a computational problem it can be phrased as follows. Given a {\plant} $\M$, a reward structure $(r,\gamma)$, and an unsafe set $\B$, find an optimal safe switching policy (if one exists).

\section{Synthesizing Safe Policies}

In this section we discuss how we might synthesize safe switching policies. We begin by examining standard approaches used in RTA design (Section~\ref{sec:basicRS}) and we identify through examples where they may fall short. We then present our approach to the problem in Section~\ref{sec:shaping}. Typical approaches to synthesizing optimal control policies like static analysis methods for MDPs or reinforcement learning, synthesize a switching policy that maximizes reward but do not guarantee its safety. Our approach is to change the reward structure so that an optimal policy for the {\plant} in the new reward structure, is an optimal, safe policy in the old reward structure, provided the value of the plant in the new reward structure is non-negative. Our reward shaping approach allows one to use off-the-shelf techniques to synthesize an optimal, safe switching policy for an RTA. Our experiments later rely on using reinforcement learning, which can be applied seamlessly to large unknown state spaces. Our formal model of a {\plant} is deterministic --- transitions associated with a state and action, result in a unique next state. However, the reward shaping approach generalizes to other transition structures as well. These are observed at the end of Section~\ref{sec:shaping}.

\subsection{Reachable and Recoverable States}
\label{sec:basicRS}

Let us fix a {\plant} $\M = \tuple{Q,q_0,A,\D}$ and an unsafe set $\B$. In this section, we will assume that the action set $A$ consists of two actions $\SC$ and $\UC$, corresponding to choosing the safe controller and untrusted controller, respectively. Let also assume that our goal is to maximize the use of the untrusted controller as much as possible. This could be captured by a reward structure where every transition of action $\UC$ gets some positive reward, while all other transitions have reward $0$.  

One standard approach to RTA design is to use the \emph{safe set lookahead policy} which is a stationary policy that uses the following decision logic: Choose the untrusted controller as long as the safety obligations are not violated. This policy can be formulated (and generalized) as follows. Let $P \subseteq Q$ be a set of {\plant} states. A \emph{$P$-lookahead policy} is a (stationary) switching policy $\pi_P$ defined as $\pi_P(q) = \UC$ if and only if $\D(q,\UC) \in P$. A safe lookahead policy is then simply a $(Q \setminus \B)$-lookahead policy. Unfortunately, such a policy does not guarantee safety.

\begin{example}
\label{ex:lookahead}
\begin{figure}
    \vspace{-0.75cm}
    \centering
    \includegraphics[width=0.35\textwidth]{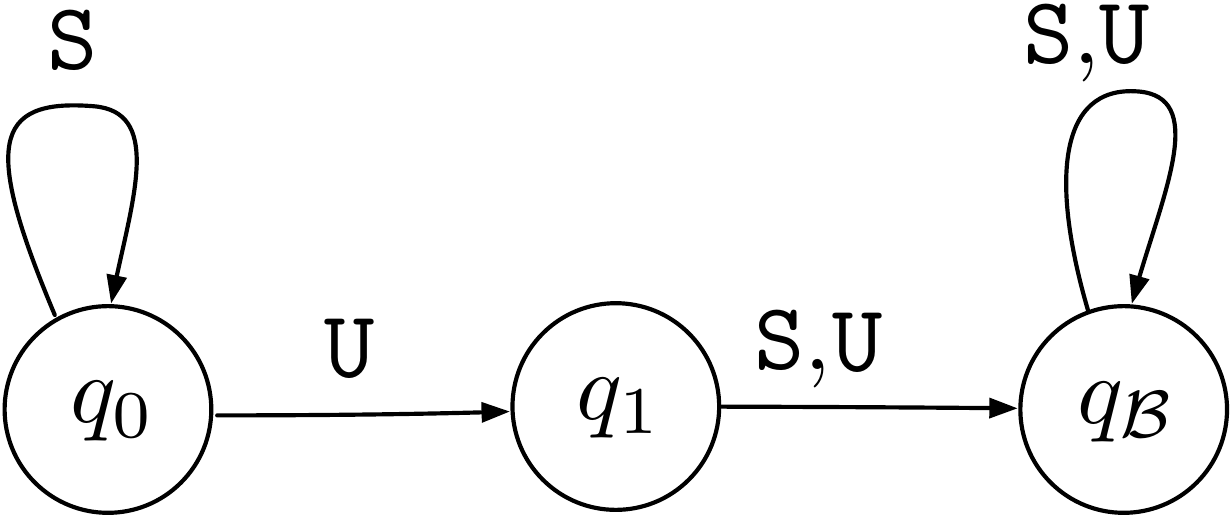} \hspace*{0.5in}
    \includegraphics[width=0.35\textwidth]{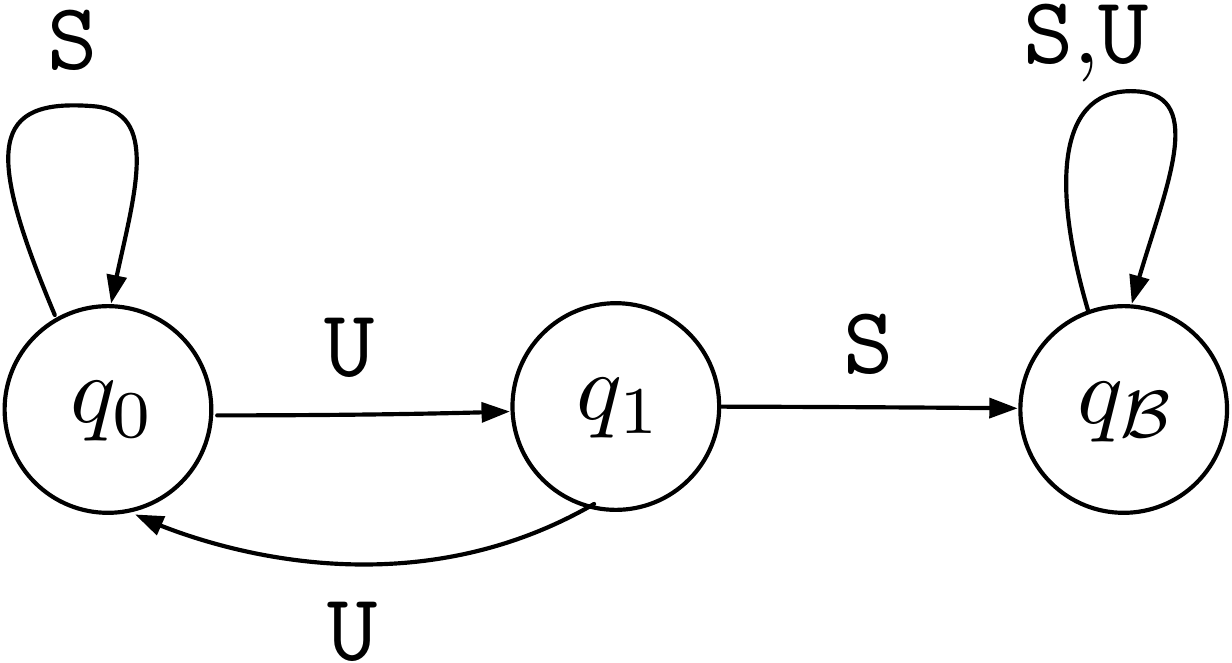}  
    \caption{\small {\em Left.} Example showing that the safe lookahead policy can be unsafe.
      {\em Right.} Example showing the sub-optimality $R$-lookahead policies for a recoverable set $R$.
    }
    \label{fig:lookahead}
\end{figure}

Consider the {\plant} shown in Fig.~\ref{fig:lookahead} (left). It has 3 states with $q_0$ being the initial state, and $\set{q_\B}$ being the unsafe set of states. The safe controller $\SC$ keeps the {\plant} in state $q_0$ if chosen in the initial state, and otherwise takes the system to the unsafe state $q_\B$ from all other states. On the other hand, the untrusted controller $\UC$ takes the system to state $q_1$ from $q_0$, and like $\SC$, takes the system to $q_\B$ from all other states. 

A safe lookahead policy would recommend using $\UC$ from $q_0$ since choosing $\UC$ from $q_0$ does not violate safety. But this will doom the {\plant} to violate safety in the very next step. A longer lookahead (instead of one step) when deciding whether to choose the untrusted controller is not a solution either --- the {\plant} in Fig.~\ref{fig:lookahead} (left) can be modified by lengthening the path from $q_1$ to $q_\B$ that delays the impending doom but does not avoid it.
\end{example}

A $P$-lookahead policy can be adapted to ensure safety by choosing an appropriate set for $P$. A set $R \subseteq Q$ will be said to be a \emph{recoverable set} for $\SC$ and unsafe set $\B$ if the following conditions hold: (a) $R \cap \B = \emptyset$, and (b) for all $q \in R$, $\D(q,\SC) \in R$. In other words, a recoverable set is a safe inductive invariant for the safe controller $\SC$ starting from a state in the recoverable set. It is easy to observe that the arbitrary union of recoverable sets is also recoverable, and so there is a unique largest (w.r.t. the subset relation) recoverable set. It is easy to show that $P$-lookahead policy is safe if $P$ is a recoverable set containing the initial state.

\begin{proposition}
\label{prop:recoverableset}
Let $\M = \tuple{Q,q_0,A,\D}$. If $R \subseteq Q$ is such that $q_0 \in R$ and $R$ is a recoverable set for $\M$ with respect to unsafe set $\B$, then the $R$-lookahead policy $\pi_{R}$ is safe.
\end{proposition}
\begin{proof}
Fix $R \subseteq Q$ and let $\rho = p_0, a_0, p_1, a_1, \ldots$ be the unique run corresponding to $\pi_R$. We should by induction that $p_i \in R$ for all $i$. For the base case, observe that $p_0 = q_0 \in R$ by assumption. For the inductive step, assume that $p_i \in R$. If $a_i = \UC$ then we know $\pi_R(p_i) = \UC$ and that means $p_{i+1} = \D(p_i,\UC) \in R$. On the other hand, if $a_i = \SC$, then by definition of $R$, $p_{i+1} = \D(p_i,\SC) \in R$. Thus, since $\rho \not\in \unsafe(\M,\B)$, $\pi_R$ is a safe policy.
\end{proof}

However $R$-lookahead policies may not be optimal as shown below.
%This can be seen by our next example.

% \begin{example}
% \label{ex:recoverable-suboptimal}
% \begin{figure}
%     \centering
%     \includegraphics[width=0.35\textwidth]{Figs/suboptima.pdf}
%        \includegraphics[width=0.35\textwidth]{Figs/suboptima.pdf}
%     \caption{\small {\em Left.} Example showing the sub-optimality $R$-lookahead policies for a recoverable set $R$. {\em Right.}
%     Example showing the sub-optimality $R$-lookahead policies for a recoverable set $R$.
%     }
%     \label{fig:recoverable}
% \end{figure}
\begin{example}
\label{ex:recoverable-suboptimal}
Consider the example {\plant} shown in Fig.~\ref{fig:lookahead} (right). There are 3 states with $q_0$ being the initial state, and the unsafe set $\B = \set{q_\B}$. It is similar to the {\plant} in Fig.~\ref{fig:lookahead} (left) and Example~\ref{ex:lookahead} with the only difference being that the action $\UC$ from $q_1$ takes the {\plant} to state $q_0$. Let us consider a reward structure $(r,\gamma)$ that rewards the use of controller $\UC$. That is, $r(q,a) = 1$ if and only if $a = \UC$ and is $0$ otherwise. 

The only (non-empty) recoverable set for this example is $R = \set{q_0}$, since $\D(q_1,\SC) \in \B$. The $R$-lookahead policy $\pi_R$ chooses $\SC$ at each step, since choosing $\UC$ at state $q_0$ takes the system out of set $R$. While $\pi_R$ is safe, it has reward $0$. On the other hand, the switching policy that chooses $\UC$ at every step is safe and has reward $1/(1-\gamma) > 0$. Therefore, $\pi_R$ is not optimal.
\end{example}

\subsection{Reward shaping}
\label{sec:shaping}

As observed in Examples~\ref{ex:lookahead} and~\ref{ex:recoverable-suboptimal}, popular approaches to designing a RTA switching policy may fail to meet either the safety or the optimality conditions that one may require. The formal model for a {\plant} we have is a special case of a Markov Decision Process (MDP)~\cite{Puterman:1994}, and so methods that identify an optimal policy for MDPs could be used here, including reinforcement learning. However, these approaches while effective in constructing an optimal policy, do not guarantee picking an optimal \emph{safe} policy. New proposals for adapting reinforcement learning to guarantee safety include approaches involving safety shields~\cite{alshiekh2018safe} (see Section~\ref{sec:related} for a comparison). Our method outlined here ``shapes the reward'' or changes the reward structure so that even using vanilla reinforcement learning that ignores the safety constraints will identify an optimal safe policy.

Consider a {\plant} $\M = \tuple{Q,q_0,A,\D}$, and a reward structure $(r,\gamma)$. For a subset $P \subseteq Q$ and $p \in \reals$, the reward function $r[P \mapsto p]$ is defined as follows.
\[
r[P \mapsto p](q,a) = \left\{ \begin{array}{ll}
                                r(q,a) & \mbox{if } q \not\in P\\
                                p & \mbox{if } q \in P
                              \end{array} \right.
\]
The main observation of this section is a result describing how to change the reward structure to compute the optimal safe policy.
\begin{theorem}
\label{thm:reward-RTA}
Let $\M = \tuple{Q,q_0,A,\D}$ be a {\plant} with $|Q| = n$. Let $(r,\gamma)$ be a reward structure for $\M$ such that $r(q,a) \geq 0$ for all $(q,a) \in Q \times A$ and let $r_{\max} = \max_{(q,a) \in Q \times A} r(q,a)$ be the maximum reward on any transition. Let $\B \subseteq Q$ be the set of unsafe states. Define $p = -\frac{r_{\max}}{\gamma^n(1-\gamma)}$ and $r' = r[\B \mapsto p]$. The following statements are true.
\begin{enumerate}[(a)]
\item There is a stationary policy $\pi$ that is optimal for $\M$ with respect to $(r',\gamma)$, i.e., $r'(\pi) = V(\M,r',\gamma)$.
\item $\sfpolicy(\M,\B) \neq \emptyset$ if and only if $V(\M,r',\gamma) \geq 0$.
\item If $V(\M,r',\gamma) \geq 0$ and $\pi$ is an optimal stationary policy for $\M$ with respect to $(r',\gamma)$ then $\pi$ is an optimal safe policy for $\M$ with respect to $(r,\gamma)$ and unsafe set $\B$.
\end{enumerate}
\end{theorem}

\begin{proof}
Observe that a {\plant} $\M$ is a special Markov Decision Process (MDP). Since optimal stationary policies always exist for MDPs with discounted rewards (Theorem 6.2.7 of~\cite{Puterman:1994}), observation (a) follows immediately.

Next, let us assume that $\sfpolicy(\M,\B) \neq \emptyset$. Suppose $\rho$ is a safe run, i.e., $\rho \not\in \unsafe(\M,\B)$. Then $r(\rho) = r'(\rho)$ because for every transition $(q,a)$ in $\rho$, $r(q,a) = r'(q,a)$ as $q \not\in \B$. For any $\pi \in \sfpolicy(\M,\B)$, since $P_\pi(\unsafe(\M,\B)) = 0$, $r(\pi) = r'(\pi)$. Finally, as $r(q,a) \geq 0$ for all $(q,a) \in Q\times A$, we have $0 \leq r(\pi) = r'(\pi) \leq V(\M,r',\gamma)$. Conversely, suppose $V(\M,r',\gamma) \geq 0$. From part (a), there is a stationary policy $\pi$ such that $r'(\pi) = V(\M,r',\gamma) \geq 0$. We will show that $\pi$ is a safe policy. For contradiction, suppose $\pi$ is unsafe. Since $\pi$ is stationary, there is exactly one run $\rho = p_0,a_0,p_1,a_1,\ldots$ such that $P_\pi(\rho) > 0$. Our assumption of $\pi$ being unsafe means that $\rho$ is unsafe, there is a $k$ such that $p_k \in \B$. Further, since $\pi$ is stationary, $k < n$.
\begin{align*}
r(\pi) = r(\rho) & = \sum_{i=0}^\infty \gamma^i r(p_i,a_i) = \sum_{i=0}^{k-1} \gamma^i r(q_i,a_i) - \gamma^k\frac{r_{max}}{\gamma^n(1-\gamma)} + \sum_{i=k+1}^\infty \gamma^i r(q_i, a_i)  \\
  & \leq r_{max} \left[\sum_{i=0}^{k-1} \gamma^i  - \frac{\gamma^{k}}{\gamma^n(1-\gamma)} + \sum_{i=k+1}^\infty \gamma^i \right]  \leq r_{max} \left[\sum_{i=0}^\infty \gamma^i  - \frac{\gamma^{k}}{\gamma^n(1-\gamma)} \right] \\
  & = r_{max} \left[\frac{1}{1-\gamma}  - \frac{\gamma^{k}}{\gamma^n(1-\gamma)} \right]  = \frac{r_{max}}{1-\gamma} \left[\left(1 - \frac{\gamma^{k}}{\gamma^n}\right) \right] < 0.
\end{align*}
The last step is because $\gamma < 1$ and $k < n$. This gives us our desired contradiction, completing our proof of (b). 

From the arguments in the previous paragraph, we have if $V(\M,r',\gamma) \geq 0$ and $\pi$ is an optimal stationary policy, then $\pi \in \sfpolicy(\M,\B)$. Since $\pi$ is stationary, the unique run $\rho$ such that $P_\pi(\rho) > 0$ must also be safe. Again from observations in the previous paragraph, we have for safe runs $\rho$, $r(\rho) = r'(\rho)$ which means $r(\pi) = r'(\pi)$. We have also observed that for any policy $\pi' \in \sfpolicy(\M,\B)$, $r(\pi') = r'(\pi') \leq r'(\pi) = r(\pi)$; the observation that $r'(\pi') \leq r'(\pi)$ follows from the fact that $\pi$ is optimal for $(r',\gamma)$. Thus, $\pi$ is an optimal safe policy for $\M$ with respect to $(r,\gamma)$ and $\B$. This completes our proof of (c).
\end{proof}

This technique of reward shaping to find an optimal safe policy can be applied to more general models than the {\plant} model considered here.

\subsubsection*{Markov Decision Processes}
A Markov Decision Process (MDP) is a generalization of the plant model. An MDP $\M$ is a tuple $\tuple{Q,q_0,A,\D}$, where $Q$, $q_0$, $A$ are like before the set of (finitely many) states, initial state, and actions, respectively. The transition function $\D$ is now \emph{probabilistic}. In other words, $\D$ is a function which on a state-action pair returns a probability distribution on states, i.e., $\D: Q\times A \to \prob(Q)$. Reward structures and unsafe states are like before. The notion of switching policies, the probability space, safe policies, reward of a policy, and the value can be generalized to the setup of MDPs in a natural manner. These formal definitions are skipped; the reader is directed to~\cite{Puterman:1994}. The goal once again is to find the best safe policy. The reward used in Theorem~\ref{thm:reward-RTA} needs to modified slightly to get a result similar to Theorem~\ref{thm:reward-RTA}. For a reward $r$, let $r'$ be the reward function given by
\[
r' = r[\B \mapsto -\frac{r_{\max}}{p\gamma^n(1-\gamma)}]
\]
where $p$ is the least probability of a path of length $n$ in $\M$. In the worst case, we can estimate $p$ as $p \geq b^n$, where $b$ is the smallest non-zero probability of any transition. Theorem~\ref{thm:reward-RTA} goes through for this $r'$. The proof is almost the same. The crucial observation that enables the proof is that an MDP $\M$ with reward structure $(r',\gamma)$, has an optimal stationary strategy which is established in Theorem 6.2.7 of~\cite{Puterman:1994}.

\section{Experimental Evaluation}
\label{sec:exp}

\subsection{Scenarios for evaluating RTA systems}
\label{sec:scenarios}
In this section, we describe a sequence  of  increasingly complex RTA problems that are used for our experimental evaluations.
A scenario is described  by at least two agents with some dynamics, with the unsafe sets defined in terms of the agents' separation distance. We  focus on the RTA system deployed on one of the agents, the {\em follower}, which is responsible for maintaining safety by switching between the safety and untrusted controllers. 
Agent dynamics are described using differential equations, which is the standard setup for  control systems. Computer implementation of these models can be seen as systems with discrete transitions in a continuous state space, which satisfies our definition of a plant.
% Definition~\ref{}, 
Such implementations are commonly used to approximate real-world dynamics.
% \sayan{(what exactly should we say about this gap?)}.
\begin{table}[!ht]
    \centering
    \caption{\small Comparison of the complexity of the experimental scenarios. Simulation D.O.F: The number of degrees of freedom necessary to describe the scenario.  Control dimension: the dimension of the  inputs generated by the controllers  in the scenario. Dynamic unsafe sets: The number of unsafe regions that must be avoided.   Layer size: number of hidden units in each layer of the trained network for the $\rlrta$.  Model parameters: Total number of trainable parameters for the $\rlrta$ used in the  scenario.}
    \vspace{0.2mm}
    \begin{tabular}{|c|c|c|c|c|c|}
    \hline 
        Scenario & \makecell{Simulation \\D.O.F.} & \makecell{Control\\ Dimension} & \makecell{Dynamic \\Unsafe Sets} & \makecell{Layer\\ Size} & \makecell{Model \\Parameters} \\ \hline
        \acc & 4 & 1 & 1 & 512 & 532,995  \\ \hline 
        \dubins & 12 & 2 & 1 & 1024 & 2,141,187   \\ \hline  
        \building & 15 & 2 & 1 & 512 & 547,331 \\ \hline 
        \groundCollision & 21 & 3 & 1 & 512 & 644,611 \\ \hline  
        {\sf Fleet\/} & 30 & 8 & 5 & 512 & 645,649\\ \hline
    \end{tabular}
\end{table}

\paragraph{Adaptive Cruise Control (\acc).}
This scenario involves a leader vehicle moving at  constant speed  and a follower trying to maintain a safe separation.  Similar ACC models have been studied elsewhere~\cite{fainekos2012verification,LoosPN11,ACC-xenofon,FanQM18}.
%https://ieeexplore.ieee.org/abstract/document/7963404
% https://link.springer.com/chapter/10.1007/978-3-642-21437-0_6
% 
% https://ieeexplore.ieee.org/stamp/stamp.jsp?arnumber=8272345
%
An engineer using RTA could test an experimental ACC controller ($\UC$) while using a well-tested or legacy design as the safety controller.
In our model, both vehicles have double-integrator dynamics. 
%(details are given in the Appendix).
% ~\ref{app:models}).
The state of the leader ($\ell$) is given by
$q^{\lead} = [x^{\lead}, v^{\lead}]^{\top}$,
where $x^{\lead}$ is the position and $v^{\lead}$ is the velocity. The state of the follower is $q^{\ego} = [x^{\ego}, v^{\ego}]^{\top}$.
The leader moves with some constant speed and the follower can apply an acceleration $a \in [-a_{\max}, a_{\max}]$ ($a_{\max} > 0$) as  chosen by either of the controllers. 
%the unverified controller function $u_{\UC}$ or the safety controller function $u_{\SC}$.
Both the untrusted and safety controllers try  to track a point at a distance $d$ behind the leader, but they use different strategies: 
The untrusted controller is a simple bang-bang controller that applies maximum acceleration $u_\UC$ if it is behind the goal and maximum deceleration otherwise.
The safety controller applies acceleration $u_\SC$ proportional  to the positional and velocity errors:
\begin{equation*}
    u_\UC(q^{\ego}, q^{\lead})  = \begin{cases}
    -a_{\max} & x^{\ego} > x^{\lead}-d \\
    a_{\max} & x^{\ego} \leq x^{\lead}-d 
    \end{cases} \quad \text{ and } \quad
    u_\SC(q^{\ego}, q^{\lead}) = \begin{bmatrix}
    k_1 & k_2
    \end{bmatrix} \begin{bmatrix}
        (x^{\lead} - d) - x^{\ego} \\
        v^{\lead} - v^{\ego}
    \end{bmatrix},
\end{equation*}
where $k_1$ and $k_2$ are some positive gains.
A safety violation occurs if the follower is within a distance $c < d$ of the leader; that is, the unsafe set $\B = \{q^{\ego} \ | \  x^{\ego} \in Ball_c(x^\lead) \}$, where $Ball_c(x^\lead)$ is the ball of radius $c$ centered at the leader. 

\paragraph{Dubin's Vehicle (\dubins).}
In this scenario, both the leader and the follower vehicles follow nonlinear $\dubins$ dynamics.
%~\cite{frazzoli?}. \mahesh{missing reference}\chris{There is a Frazzoli citation in our references, but I'm not sure it's the one that was intended here.}
%(details given in the Appendix).
% Appendix~\ref{app:models}). 
This setup is used to study platooning protocols for ground and air vehicles~\cite{ribichini2003efficient,dunlap2022run,dunlap2021comparing,eller2013test,wadley2013development}.
The leader's state is given by 
%$q^{\lead} = [x^{\lead}, y^{\lead}, \yaw^{\lead}, v^{\lead}]^{\top}$, where 
the $xy$-position $(x^{\lead}, y^{\lead})$,
the heading $\yaw^{\lead}$, and
the speed $v^{\lead}$. 
The leader moves in different  circular paths defined by the radius $r$ and the speed.
%$\cstv^{\lead}$.
%
The follower's state is  $q^{\ego} = [x^{\ego}, y^{\ego}, \yaw^{\ego}, v^{\ego}]^{\top}$.
The control input for the follower is $u = [\yawInput, a]^{\top}$, where $\yawInput$ is the heading rate and $a$ is the acceleration.
It tries to track some state ``behind'' the leader defind as $[x^{\lead}-d \cos{\yaw^{\lead}}, y^{\lead}-d \sin{\yaw^{\lead}}, \yaw^{\lead}, v^{\lead}]$, $d > 0$.
The untrusted controller ($\UC$), adapted from previous studies~\cite{fan2020fast,kanayama1990stable},
% \mahesh{references not compiling}, 
implements a Lyapunov-based feedback controller and the safety controller ($\SC$) is similar to $\UC$ except that it does not accelerate to catch-up. 
Again, the goal of the RTA is to switch between the two controllers to prevent the follower from entering $Ball_c(x^\lead)$.

% untrusted and safety controllers such that the follower does not collide with the leader.
% A collision between the leader and the follower occurs if $\|(x^{\lead}, y^{\lead}) - (x^{\ego}, y^{\ego})\| \leq c$, $c < d$.
% The unsafe set is then defined as $\unsafe = \{q^{\ego} | \|(x^{\lead}, y^{\lead}) - (x^{\ego}, y^{\ego})\| \leq c \}$.

\paragraph{Dubins with Obstacles ($\building$).}
This scenario is the same as $\dubins$ except that there is a rectangular obstacle in the follower's path, and the safety controller ($\SC$), when activated, tracks a different reference point in order to avoid collision with the obstacles. The unsafe set, accordingly, includes both collision with the leader and the obstacles.

\paragraph{Ground and Air Collision Avoidance ($\groundCollision$).}
The leader and the wingman fly through  3D space with multiple ground and building obstacles. The six dimensional state of the leader is  $q^{\lead} = [x^{\lead}, y^{\lead}, z^{\lead}, \yaw^{\lead}, \pitch^{\lead}, v^{\lead}]^{\top}$,
where $z^{\lead}$ is the altitude and $\pitch^{\lead}$ is the pitch. The leader follows certain circular or elliptical orbits in  space. Dynamics are given in the Appendix.

%
% The leader is moving in a circular path of radius $r$ at constant height $h$ and constant velocity $\cstv^{\lead}$, meaning the pitch of the leader is always $0$.
%
%Similarly, the state of the follower is given by $q^{\ego} = [x^{\ego}, y^{\ego}, z^{\ego}, \yaw^{\ego}, \pitch^{\ego}, v^{\ego}]^{\top}$.
%Both leader and follower aircraft follow the $\air$ dynamics described in Appendix~\ref{app:models}.
The input for the follower is  $u = [\yawInput, \pitchInput, a]$, where $\yawInput$ is the heading rate, $\pitchInput$ is the pitch rate, and $a$ is the acceleration.
The follower  attempts  to track a reference point that is $d$-behind and $\delta$-below the leader. 
%
% $[x^{\lead} - d \cos{\yaw^{\lead}}, y^{\lead} - d \sin{\yaw^{\lead}}, z^{\lead} - \delta, \yaw^{\lead}, \pitch^{\lead}, v^{\lead}]$, $d > 0$ and $0 < \delta < h$.
The untrusted controller is similar to the ones used in $\dubins$ except that it now includes a $\pitchInput$-component: 
\begin{equation}
    u_{\UC}(q^{\ego}, q^{\ell}) = \begin{bmatrix}
        \yawInput \\ \pitchInput \\ a^c
    \end{bmatrix} = \begin{bmatrix}
        \yawInput^{\lead} + v^{\ell}(k_1 \err_y + k_2 \sin{(\yaw^{\ell} - \yaw^{\ego})} \\
        k_4 (\pitch^{\rf} - \pitch^{\ego}) \\
        k_3(v^{\ell} - v^{\ego})
    \end{bmatrix},
\end{equation}
where $\err_y$ is the $y$-error,
$\yawInput^{\rf}$ is the heading rate of the leader, and
$\pitch^{\rf}$ is the angle from the follower's current position to  the reference point. 
%A full description of this controller is provided in Appendix~\ref{app:controllers}.
The safety controller is similar to the untrusted controller, but it instead tracks the state $[x^{\lead} - d \cos{\yaw^{\lead}}, y^{\lead} - d \sin{\yaw^{\lead}}, z^{\lead}, \yaw^{\lead}, \pitch^{\lead}, v^{\lead}]$ and, similar to the $\dubins$ and $\building$ safety controllers, sets the acceleration to $0$.
% \sayan{Say in English how it is different.}
% 
Once again, The goal of the RTA is to switch between the untrusted and safety controllers such that the follower does not collide with the leader or any of the $n$ polytopic obstacles.
%
% $\poly(A_i, b_i) = \{p \in \reals^{3}| A_i p \leq b_i\}$ where $A_i \in \reals^{3 \times 3}$ and $b_i \in \reals^{3}$.
% The unsafe set is $\unsafe = \{q^{\ego} | (x,y) \in Ball_c(x^\lead, y^\lead) \cup \bigcup_{i\in \{1, \dots, n\}} \poly(A_i, b_i)) \}$.
%\kristina{Should we include a picture of an example 3D scenario here?}
An example of $\air$ can be seen in Fig.~\ref{fig:cagas}.

\paragraph{Multi-agent Dubins ({\sf {\em Fleet}}).}
The final scenario is a generalization of the $\building$ scenario in which the single wingman agent is replaced by a group of four follower aircraft.  The followers' tracking points are set in a V pattern, with the safety controller tracking a point farther away from the leader.  For each agent, the unsafe set consists of the ball around the leader, the ball around each other agent, and the obstacles.  Here, the RTA is a single centralized\footnote{Decentralized RTA strategies can also be developed in our framework and this will be the subject of future works.} decision-making module that determines the switching logic of all the agents simultaneously. 
\begin{figure}[ht!]
    \centering
\vspace{-0.7cm}
    \includegraphics[width=0.4\textwidth, trim=0 0 0 5, clip]{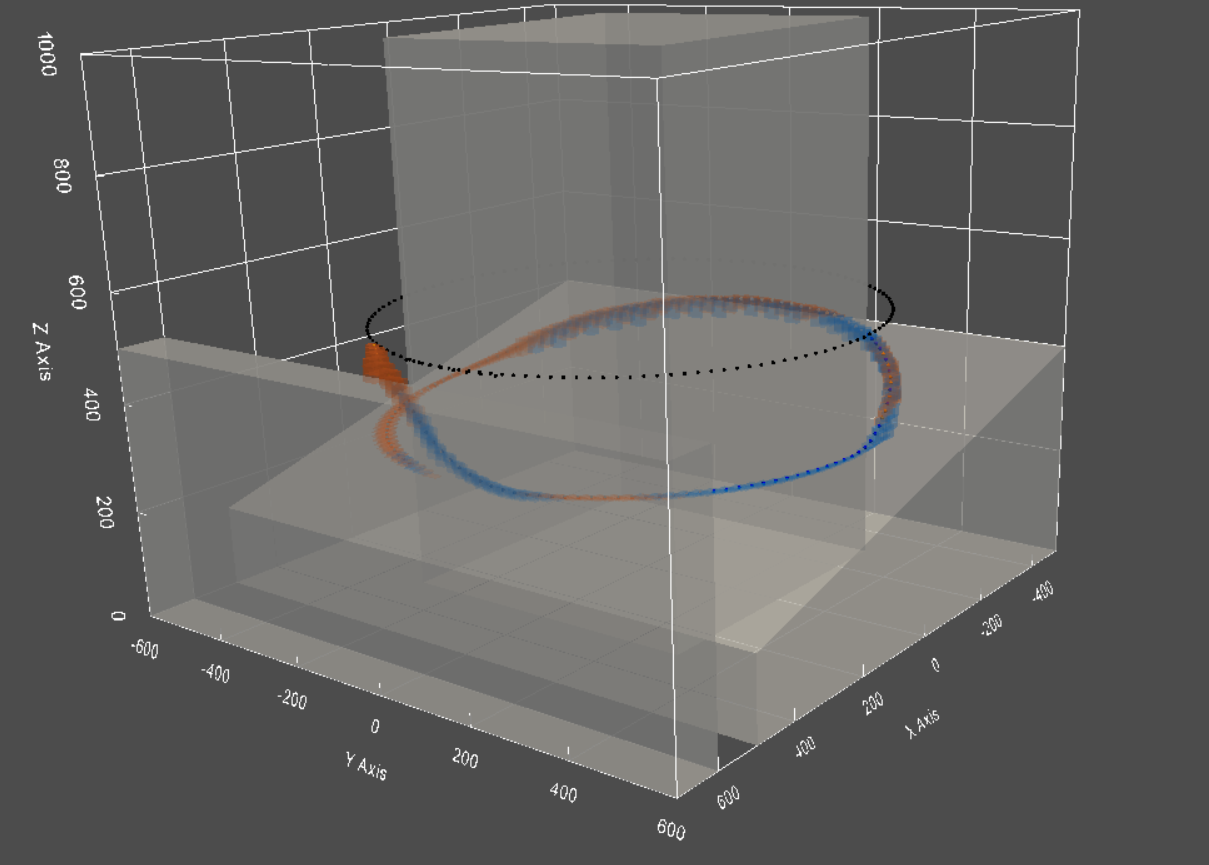} \hspace*{0.5in}
    \includegraphics[scale=0.115,width=0.4\textwidth, trim=0 165 0 5, clip]{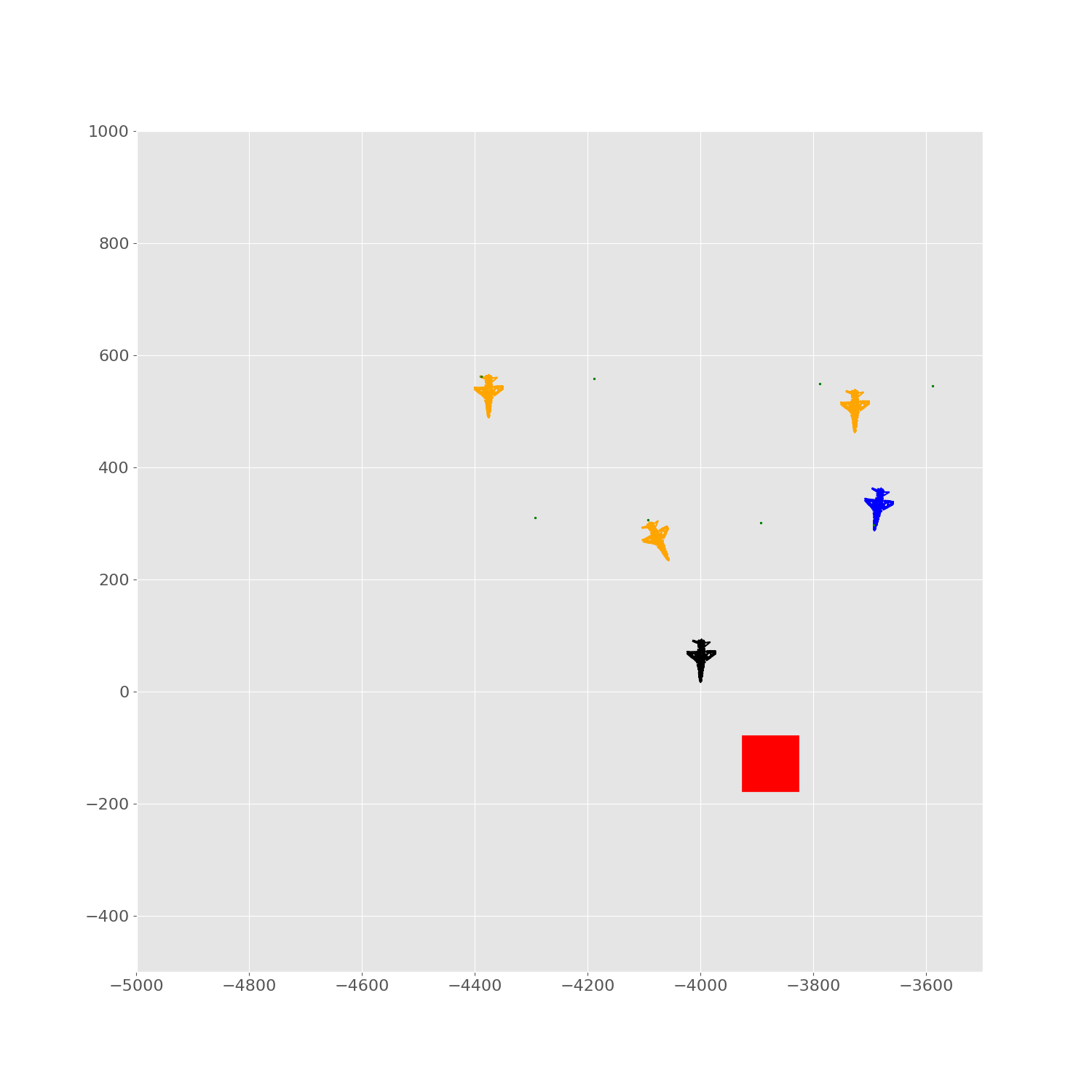}
    \caption{\small {\em Left.} An example $\groundCollision$ scenario execution. Leader trajectories (back dots), follower's reachable states (red/blue), and obstacles (gray). {\em Right}. An example {\sf Fleet\/} scenario. Leader (black) followed by four followers (blue/orange) with an obstacle (red).}
    % \sayan{Lets drop the fig on the right. It comes across as too simple. Instead we could show an interesting run of the 1-d example which shows something interesting about utilization / violation of safety.. or the switching boundaries / recoverable sets. }
    % Right: An example of $\building$, with a leader trajectory (black dots), follower trajectory (green dots), untrusted controller trajectory (yellow dots), safe trajectory (blue dots), and building (red rectangle).}
    \label{fig:cagas}
\end{figure}

\vspace{-0.75cm}
\subsection{RTA Implementations}

\paragraph{$\reachrta$ and $\simrta$ (Baselines).} 
In Section~\ref{sec:basicRS}, we discussed two popular approaches for implementing RTAs based on lookaheads. The exact current state $q$ of the system is usually not  available for computing the lookahead. Instead, a state estimator gives a set of states  $\hat{Q} \subseteq Q$ (or a mean  estimate plus error) based on sensor data. The lookahead computation $\D(\hat{Q})$ propagates this set and checks intersections with a set $P$. For a longer lookahead time $T$, the intersection with $P$ is  checked with the reachable states $\reach^\UC(\hat{Q},T)$ from $\hat{Q}$. 
%\sayan{Reachable states defined earlier? Use right notation.}
%
We will use two  implementations of this lookahead strategy: (a) $\reachrta$  uses the Verse library~\cite{Verse23} for over-approximating the reachable sets, and (b) $\simrta$  uses simulations from $\hat{Q}$ to estimate the reachable sets. The latter is lightweight but does not give an over-approximation of the reachable sets.  
%
%For $P$-lookahead policies, we also have to choose the set $P$. 
As discussed in  Section~\ref{sec:basicRS}, using a recoverable set for the set $P$ guarantees safety but not optimality, and using the safe set $Q\setminus \B$ does not even guarantee safety. In general, computing the recoverable set can be hard, and we use the safe set; however, for the $\acc$ scenarios, we have been able to compute an analytical solution for the recoverable set, and that is what we use in the experiments.

% \sayan{Measurement error as motivation for reachability/sim RTA.}
 
%  $\reachrta$ chooses between $\SC$ and $\UC$ by \sayan{write in 1-2 sentences about implementation principle, use Reach() notation, refer back to Defs as needed. Give 1 sentence about software tools used.}
% % 
% $\simrta$ is similar to $\reachrta$, but it computes a single trajectory and checks that no points along that trajectory are in the unrecoverable set.

\paragraph{$\rlrta$.}
In order to train the $\rlrta$, we used the SafeRL reinforcement learning framework~\cite{Hobbs-SafeRL22}, an extension of OpenAI gym~\cite{brockman2016openai}, to gather observation-action pairs.  Episodes lasted either for a fixed number of time steps (150 steps for the $\acc$ scenario, and 400 for the others) or until the a safety constraint was violated.
The Ray reinforcement learning library was then used to train models based on the observations-action pairs  using proximal policy optimization (PPO)~\cite{PPO}.  

The models are neural nets consisting of two fully connected hidden layers with 512 units each, except for the \dubins~example which required 1024 units to achieve a safe policy.  The models were trained for 20 training iterations, each consisting of  a batch of 200,000 simulation time steps.  The models were then evaluated after each training iteration, and from the safe models, the model that performed best at the objective was selected.  Four learning hyperparameters were optimized using Ray's hyperparameter tuner, yielding values of 0.995 for $\lambda$ and $\gamma$, 150 for the number of gradient descent iterations per training step, and 75 for the minibatch size.
Models were trained with two types of reward functions.  In the first, a reward of 1 was given in each time step that the unverified controller was used.  In the second, a reward of 1 was given when the agent was in a specified rejoin region.  We use the reward shaping strategy from Proposition 2 to generate the penalty for violating safety constraints\footnote{Having an explicit form for the necessary penalty is of great value (e.g., in determining the necessary data type to use in training).}.  
%Because these simulations are not in a finite state space, w
We replace the parameter $n$ with $T_{max}$, the maximum number of time steps in a given episode.

\subsection{Effectiveness and Scalability}
\label{sec:scale}

We compare several  RTA strategies across three variations of the single agent scenarios of Section~\ref{sec:scenarios}, created by altering the relative initial position and velocity of the follower. 
These experiments were run on a 2.4 GHz Quad-Core Intel i5 processor with 8 GB RAM.
Variation 1 is the variation that $\rlrta$  was trained on; Variations 2 and 3 make the initial conditions of the follower different from the training scenarios.
%, 
%where the initial position is offset by some $+\delta$ in variation 2 and some $-\delta$ in variation 3.
%}
% Sayan: maybe too much detail
%
For each RTA and scenario, we report
(a) the average running time (RT) of the RTA switch decision in milliseconds, 
(b) the percentage of time the untrusted controller is used $\UC \%$ over the entire run, and 
(c) the minimum time to collision (TTC) as a measure of safety~\cite{ttc} with the unsafe set  $\B$ over the whole run. A negative TTC indicates violation of safety.
The results are reported in Table~\ref{tab:inc-exp}.

{\em Safety violations are possible but rare.} First, even in the safe scenarios, the TTC is small. This is because these scenarios were constructed to allow the follower to track the leader up-close and at relatively high speeds. 
Second, it is common to find scenarios where  $\simrta$ and $\reachrta$ violate safety. This is not surprising given Example~\ref{ex:lookahead}, and these scenarios give concrete and more realistic instantiations.
$\rlrta$ is always safe in these scenarios, even though the varied initial conditions were not part of the training set.
% and in this case, the initial conditions for the follower were well outside the training set.
%
% \sayan{Something in Acc Var 2 makes ReachRTA so unsafe? It is a huge outlier.}
% \sayan{I'm concerned that RLRTA does not look safer than the others here.}

{\em $\rlrta$ can improve utilization of the untrusted controller over $\reachrta$ and $\simrta$, often by more than 50\%.} The improvements are more drastic from $\reachrta$, which is considered to be state-of-the-art with regards to safety. These results illustrate that our approach addresses the well-known conservativeness problem of RTAs.  In the {\building} and {\sf Fleet} scenarios, the $\reachrta$ was able to achieve higher utilization, indicating that, in those cases, the training of $\rlrta$ yielded a safe, but not optimal, policy, indicating that improvements to the training process and model definition are possible.
The improvements in $\acc$ are larger than the improvements in the other scenarios as the reachable set computation is more conservative with respect to the unsafe sets than in the other scenarios.
% much larger recoverable sets than the conservative ones used for $\reachrta$.
% \sayan{How do we know this?}
% The improvements are  less in $\air$ because safety sets used for $\reachrta$ being less conservative with respect to the recoverable set.
% \sayan{What does this mean? For $\air$ no recoverable set was computed right? RLRTA is not explicitly computing recoverable sets, just using reward shaping.}
% \sayan{Explain why improvements are high in Acc but not so high in others.}

% \sayan{Are we using recoverable sets here for simRTA and reachRTA? Then why is ReachRTA unsafe in ACC.2? We should so explain why RLRTA became unsafe. Is this a scenario outside of the training set?}

{\em Running time increases with plant dimension, but $\rlrta$ scales best}. 
The plant models for $\acc$, $\dubins$, and $\air$ are 4-, 8-, and 12-dimensional, respectively.
For a given scenario, $\rlrta$ is generally faster than  $\reachrta$, and it can be 40 times faster for large models. At runtime, $\rlrta$ switching policy $\pi$ involves running inference on a neural network, whereas $\reachrta$ involves computing the reachable states of a nonlinear model. As such, the runtime of $\rlrta$ increases as the required complexity of the policy increases, which is not necessarily correlated with increased scenario complexity.

% For the $\acc$ scenarios, we
% \sayan{Not clear if anything general can be said about this.}
% \sayan{Observations about utilization and safety.}

\vspace{-0.75cm}
\begin{table}[!h]
\centering
\caption{\small Running time (RT), time to collision (TTC) with unsafe set $\B$, and untrusted controller usage ($\UC$\%) for different RTA strategies and across scenarios.}
\vspace{0.2mm}
% These results are recorded on a workstation with a Intel\textregistered Xeon E5-2630 8 core CPU \haoqing{??}.}
\label{tab:inc-exp}
\begin{tabular}{|c|l|c|c|c|c|c|c|c|c|c|}
\hline
\scriptsize 
    &  & \multicolumn{3}{c|}{$\simrta$}  & \multicolumn{3}{c|}{$\reachrta$}   & \multicolumn{3}{c|}{$\rlrta$}  \\ \hline
    \multicolumn{2}{|c|}{}  &  RT (ms) & \makecell{ TTC \\ (ms) } & $\UC 
 \%$ & RT (ms) & \makecell{ TTC \\ (ms) } & $\UC \%$  & RT (ms) & \makecell{ TTC \\ (ms) } & $\UC \%$  \\ \hline
    \acc
                           % & $\theta$, US            & time & Safety & UC usage & 4.92 & 0.791 & 16.10 & 1.54 & 0.807 & 96.64 \\
                           & Var 1            & 0.92 & 335 & 69.79 & 4.88 & 123 & 12.08 & 3.58 & 680 & 54.3 \\
                           & Var 2   & 0.96 & -82 & 65.77 & 10.99 & -19x$10^3$ & 14.76 & 3.71 & 617 & 64.9  \\
                           & Var 3   & 0.90 & -576 & 67.11 & 8.35 & 617 & 18.79 & 3.85 & 630 & 72.2 \\
    \hline
    % \multirow{1.8}{5.7em}{\dubins}
    \dubins
                           & Var 1           & 6.84 & 122 & 69 & 39.91 & 153 & 35.58 & 3.33 & 86 & 53.3   \\
                           & Var 2 & 7.07 & 111 & 67 & 37.54 & 154 & 19.93  & 3.35 & 90 & 49.0 \\
                           & Var 3 & 7.24 & 115 & 65 & 37.95 & 179 & 23.92 & 3.63 & 104 & 48.3  \\
    \hline
    % \multirow{1.8}{5.7em}{\building}
    \building
                           & Var 1           & 26.72 & 47 & 51.62  & 97.20 & 32 & 77.87 & 3.37 & 194 & 52.75   \\
                           & Var 2           & 27.60 & 62 & 51.32  & 93.11 & 76 & 76.40 & 3.70 & 179 & 48.8   \\
                           & Var 3           & 28.17 & 34 & 51.62  & 91.53 & 40 & 76.40 & 3.38 & 208 & 54.0  \\
    \hline
    \groundCollision
                           & Var 1          & 361 & 373 & 66.15  & 192.84 & 506 & 64.61 & 3.78 & 639 & 89.0 \\
                           & Var 2 & 375 & 373 & 66.15  & 178.46 & 482 & 64.61 & 3.28 & 647 & 86 \\
                           & Var 3 & 362 & 374 & 66.15  & 183.78 & 530 & 66.15 & 3.31 & 667 & 88.8 \\
    \hline
    {\sf Fleet\/}
                           & Var 1          &  &  &   & 876 & 0.95 & 53.8 & 6.24 & 212 & 46.7 \\
    \hline
\end{tabular}
\end{table}

% \sayan{Scenarios where reach RTA}
% \sayan{Possible changes: (1) Discuss training and testing in the introduction and beginning of Section 4, (2) RLRTA will always be safe on the training set, but not on different scenarios outside of training. (3) Pull out the Acc Var 3 and Air Var 3 from Table 1, and include that data in table 2.}

\vspace{-0.75cm}
\subsection{Different Rewards and Goals} 
\label{sec:rewards}

Our reward shaping approach for RTA enables us to start with different rewards, say $r_1$ and $r_2$, on the same plant model $\M$, and shape them both to make the optimal policies safe. This enables designers to play with different control objectives and goals and apply the same method for uniformly  achieving safety. 

We report on experiments with two different rewards: (a) $r$, which is a baseline reward that rewards staying within a zone around the reference point (behind the leader), and 
(b) $r_\UC$, which rewards  using the untrusted controlled $\UC$.
% Reward $r_{join}$ favors achieving the goal of rejoining the leader while $r_{\UC}$ favors experimentation with the untrusted controller. 
% \kristina{Reward $r$ is the baseline reward and
% % and favors achieving the goal of rejoining the leader while 
% $r_{\UC}$ favors experimentation with the untrusted controller. 
% }
%
We test the learned switching policies in 100 scenarios with randomized parameters.
% In these experiments, we also test the  learned switching policies on scenarios that are different from the training scenarios.
The results from the $\rlrta$ trained using these two different rewards is compared to the baseline results from $\reachrta$.
% \sayan{If this training-testing setup is discussed earlier, then it could be (re)moved fromt his section.}
The following metrics are collected and averaged over the 100 trials run:
% (a) $R\%$ is the percentage of time over the scenario that the follower remains in the rejoin region
(a) $\UC$\%, as before, is the percentage of time the untrusted controller is used;
(b) Mean dist is the mean distance of the follower to leader (or the closest aircraft, in the multi-agent case); and
(c) Fail\%  is the percentage of scenarios where safety is violated.

In the multi-agent case, reward $r$ RTA yielded slightly higher utilization than the $r_{UC}$ RTA. Similar behavior occurred in the $\building$ scenario, where there was little difference between the utilization achieved by the different RTAs.  This indicates that little optimization of the rewards was possible in these scenarios and that the learned policy is almost entirely determined by the safety requirement, allowing $\reachrta$ to outperform it on these tasks.
\vspace{-0.75cm}
\begin{table}[h!]
\centering
\small
\caption{\small Running time, safety violations, and untrusted controller usage for different RTA strategies and across scenarios. }
% \kristina{Dropped the rejoin percentage in this table. Saved the previous numbers in another table.}}
\vspace{0.2mm}
% These results are recorded on a workstation with a Intel\textregistered Xeon E5-2630 8 core CPU \haoqing{??}.}
\label{tab:inc-exp}
\begin{tabular}{|c|c|c|c|c|c|c|c|}
\hline
\multirow{2}{5.7em}{Scenario} & \multirow{2}{5.7em}{Reward} & \multicolumn{3}{c|}{$\rlrta$} & \multicolumn{3}{c|}{$\reachrta$} \\ \cline{3-8}
&  & $\UC$\% & Mean dist & Fail\%  & $\UC$\% & Mean dist & Fail\% \\ \hline
% \multirow{2}{5.7em}{\acc} & UC Usage & Follower position & Follower velocity & & & \\
{\acc} & $r_{\UC}$ & 65.3 & 12.3 & 0 & 19.5& 19.5 & 0 \\
& $r$  & 5.6 & 11.6 & 0 & &  & \\
\hline
{\dubins} & $r_{\UC}$ & 56.6 & 176.9 & 0 & 7.44& 99.00 & 0 \\
& $r$  & 13.3 & 222.8 & 0 & &  & \\
\hline
{\building} & $r_{\UC}$  & 52.5 & 376.6 & 0 & 77.86 & 338.55 & 0 \\
& $r$  & 48.2 & 372.7 & 0 & &  &  \\
\hline
{\air} & $r_{\UC}$ & 89.0 & 391.1 & 0  & 48.07 & 651.22 & 0 \\
& $r$  & 59.5 & 387.5 & 0 & &  & \\
\hline
{\sf Fleet\/} & $r_{\UC}$ & 46.8 & 275.6 & 0  & 53.8 & 404.2 & 0 \\
& $r$  & 47.2 & 275.8 & 0 & &  & \\
\hline
\end{tabular}
\end{table}
% \kristina{Some other notes:
% Incentivizing the untrusted controller always improved the untrusted controller usage.
% }

% \kristina{I believe that the reachrta rejoin percentage is 0 because the untrusted controller was not used enough to let it catch up, as well as the fact that the rejoin region in \air is so much smaller than the rejoin for \dubins.}

\vspace{-0.75cm}
\section{Conclusion and Discussions}
\label{sec:conc}

%\mahesh{Shortened summary of paper.}
In this paper we presented a novel formulation of the RTA problem, namely, as a problem of identifying a policy that maximizes certain rewards (like the use of the untrusted controller) while guaranteeing safety of the system. We showed how this problem can be reduced to synthesizing policies that are optimal with respect to only one objective function through reward shaping. The reduction allows one to use scalable techniques like reinforcement learning that are \emph{model-free} to solve the RTA problem. This approach overcomes weaknesses of traditional approached to solving the RTA problem that may either fail to guarantee safety or not be optimal.
We have implemented the reward shaping-based RTA design procedure using SafeRL and OpenAI gym frameworks, and our experiments show that this $\rlrta$ can indeed find safe switching policies  in complex scenarios. In most scenarios, $\rlrta$  improves the utilization of the untrusted controller significantly over reachability-based RTAs. These results show promise and suggest new directions for research for finding reward shaping strategies for  other hard constraints (like reachability)  in RTA systems.

While safety is an important requirement, it is not the only type of hard constraint that a RTA switching policy may need to satisfy. For example, we require the switching policy to reach a certain goal state while optimizing certain rewards. Or perhaps, we may require the RTA to ensure that a goal state is reached while avoiding some unsafe states, and at the same time maximize the use of the untrusted controller. In general, the RTA problem demands finding an optimal switching policy (like the one that uses the untrusted controller the most) from amongst those that satisfy a logical property perhaps expressed in some temporal logic. Solving this more general problem is unfortunately challenging.

\begin{wrapfigure}[8]{r}{0.3\textwidth}
    \centering
    \includegraphics[width=0.23\textwidth]{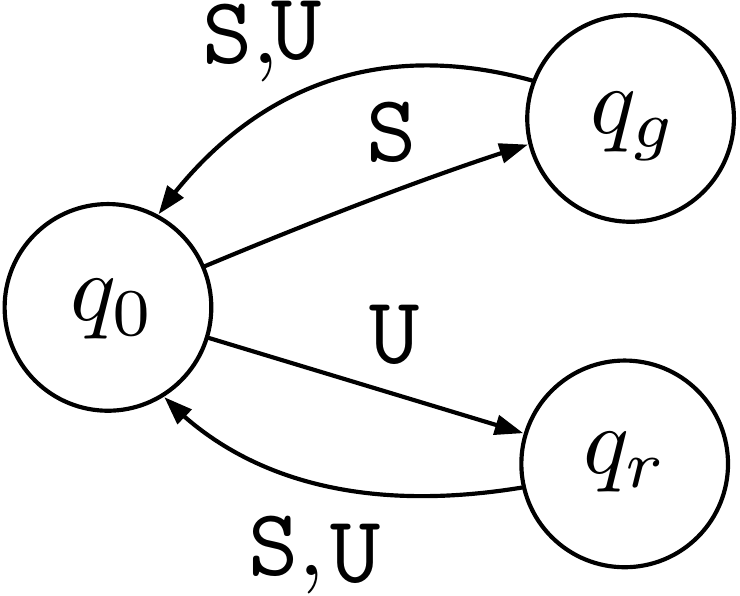}
    %\caption{Reachability is hard.}
    %\label{fig:opt}
\end{wrapfigure}
To illustrate this point, consider the {\plant} $\M$ shown on the right.
%in Fig.~\ref{fig:opt}. 
It has three states --- the initial state $q_0$, a goal state $q_g$ and a reward state $q_r$. The safe controller $\SC$ takes $\M$ from $q_0$ to $q_g$, and brings it back to $q_0$ from both $q_r$ and $q_g$. On the other hand, the untrusted controller $\UC$ takes $\M$ from $q_0$ to $q_r$ and brings it back to $q_0$ from both $q_r$ and $q_g$. Consider a reward structure $(r,\gamma)$, where $r(q,a) = 1$ if $q = q_r$ and is $0$ otherwise; in other words, reaching $q_r$ gives a reward and nothing else. Suppose our goal is to find a policy, from amongst those that reach goal state $q_g$, that maximizes the reward earned. Observe that for every $k$, there is a policy that surely reaches $q_g$ and earns reward $\frac{1}{1-\gamma} - \gamma^{2k}$ --- choose $\UC$ for the first $2k$ steps (i.e., go from $q_0$ to $q_r$ and back $k$ times), then choose $\SC$, and afterwards choose $\UC$ for the rest of the steps. Notice that these policies while deterministic, are not stationary. In fact, stationary policies that guarantee reaching $q_g$ (like choosing $\SC$ always) earn reward $0$. 

Using $\mathsf{ReachSP}(\M, \set{q_g})$ to denote the set of all switching policies that guarantee reaching the goal state $q_g$, we can observe that
\[
\forall k.\ \frac{1}{1-\gamma} - \gamma^{2k} \leq \sup_{\pi \in \mathsf{ReachSP}} r(\pi) \qquad \mbox{and} \qquad \sup_{\pi \in \mathsf{ReachSP}} r(\pi) \leq \frac{1}{1-\gamma}.
\]
This means that $\sup_{\pi \in \mathsf{ReachSP}} r(\pi) = \frac{1}{1-\gamma} = v$. But there is no switching policy that achieves this supremum value $v$! Thus, there is no optimal scheduler that guarantees the reachability objective. We have already seen that for every $\epsilon$, there is a policy that guarantees visiting $q_g$ and has reward $\frac{1}{1-\gamma} - \epsilon$. The policy described uses memory but is deterministic. One can show that we can also come up with a memoryless, but randomized policy that is $\epsilon$-optimal; the policy picks $\SC$ or $\UC$ from $q_0$ based on a distribution defined based on $\epsilon$. The existence of memoryless, randomized policies that are $\epsilon$-optimal can be proved for any {\plant} and regular property. How one can automatically design such a policy is an open problem for future investigation.

\bibliographystyle{splncs04}
\bibliography{sayan1.bib,references.bib}

\begin{thebibliography}{10}
\providecommand{\url}[1]{\texttt{#1}}
\providecommand{\urlprefix}{URL }
\providecommand{\doi}[1]{https://doi.org/#1}

\bibitem{RTA-AFC2010}
Aiello, A., Berryman, J., Grohs, J., Schierman, J.: Run-time assurance for
  advanced flight-critical control systems. In: Proc. AIAA Guidance,
  Navigation, and Control Conference, AIAA 2010-8041, Toronto, Ontario Canada,
  Aug., 2010

\bibitem{alshiekh2018safe}
Alshiekh, M., Bloem, R., Ehlers, R., K{\"o}nighofer, B., Niekum, S., Topcu, U.:
  Safe reinforcement learning via shielding. In: Proceedings of the AAAI
  Conference on Artificial Intelligence. vol.~32 (2018)

\bibitem{alur_framework_2021}
Alur, R., Bansal, S., Bastani, O., Jothimurugan, K.: A framework for
  transforming specifications in reinforcement learning. Springer Festschrift
  in honor of Prof. Tom Henzinger  (October 2021).
  \doi{10.48550/arXiv.2111.00272}, \url{https://arxiv.org/abs/2111.00272v3}

\bibitem{ames2019control}
Ames, A.D., Coogan, S., Egerstedt, M., Notomista, G., Sreenath, K., Tabuada,
  P.: Control barrier functions: Theory and applications. In: 2019 18th
  European control conference (ECC). pp. 3420--3431. IEEE (2019)

\bibitem{bak2009system}
Bak, S., Chivukula, D.K., Adekunle, O., Sun, M., Caccamo, M., Sha, L.: The
  system-level simplex architecture for improved real-time embedded system
  safety. In: 2009 15th IEEE Real-Time and Embedded Technology and Applications
  Symposium. pp. 99--107. IEEE (2009)

\bibitem{RTA-Distt2013-BakHuang}
Bak, S., Huang, Z., Abad, F.A.T., Caccamo, M.: Safety and progress for
  distributed cyber-physical systems with unreliable communication. ACM Trans.
  Embed. Comput. Syst.  \textbf{14}(4) (sep 2015). \doi{10.1145/2739046}

\bibitem{BakHAC15}
Bak, S., Huang, Z., Abad, F.A.T., Caccamo, M.: Safety and progress for
  distributed cyber-physical systems with unreliable communication. {ACM}
  Trans. Embed. Comput. Syst.  \textbf{14}(4),  76:1--76:22 (2015).
  \doi{10.1145/2739046}, \url{https://doi.org/10.1145/2739046}

\bibitem{sandbox11}
Bak, S., Manamcheri, K., Mitra, S., Caccamo, M.: Sandboxing controllers for
  cyber-physical systems. In: 2011 {IEEE/ACM} International Conference on
  Cyber-Physical Systems, {ICCPS} 2011, Chicago, Illinois, USA, 12-14 April,
  2011. pp. 3--12. {IEEE} Computer Society (2011). \doi{10.1109/ICCPS.2011.25},
  \url{https://doi.org/10.1109/ICCPS.2011.25}

\bibitem{bak2011sandboxing}
Bak, S., Manamcheri, K., Mitra, S., Caccamo, M.: Sandboxing controllers for
  cyber-physical systems. In: 2011 IEEE/ACM Second International Conference on
  Cyber-Physical Systems. pp. 3--12. IEEE (2011)

\bibitem{balakrishnan22}
Balakrishnan, A., Jaksic, S., Aguilar, E., Nickovic, D., Deshmukh, J.:
  Model-free reinforcement learning for symbolic automata-encoded objectives.
  In: Proceedings of the ACM International Conference on Hybrid Systems:
  Computation and Control. pp. 26:1--26:2 (2022)

\bibitem{bozkurt2020control}
Bozkurt, A.K., Wang, Y., Zavlanos, M.M., Pajic, M.: Control synthesis from
  linear temporal logic specifications using model-free reinforcement learning.
  In: 2020 IEEE International Conference on Robotics and Automation (ICRA). pp.
  10349--10355. IEEE (2020)

\bibitem{brockman2016openai}
Brockman, G., Cheung, V., Pettersson, L., Schneider, J., Schulman, J., Tang,
  J., Zaremba, W.: Openai gym (2016), \url{http://arxiv.org/abs/1606.01540},
  cite arxiv:1606.01540

\bibitem{SafeRL22}
Brunke, L., Greeff, M., Hall, A.W., Yuan, Z., Zhou, S., Panerati, J.,
  Schoellig, A.P.: Safe learning in robotics: From learning-based control to
  safe reinforcement learning. Annual Review of Control, Robotics, and
  Autonomous Systems  \textbf{5}(1),  411--444 (2022).
  \doi{10.1146/annurev-control-042920-020211},
  \url{https://doi.org/10.1146/annurev-control-042920-020211}

\bibitem{shieldingpomdp22}
Carr, S., Jansen, N., Junges, S., Topcu, U.: Safe reinforcement learning via
  shielding under partial observability (2022)

\bibitem{rta-cps-all}
Clark, M., Koutsoukos, X., Kumar, R., Lee, I., Pappas, G., Pike, L., Porter,
  J., Sokolsky, O.: Study on run time assurance for complex cyber physical
  systems. Tech. Rep. ADA585474, Air Force Research Lab (April 2013), available
  at \url{https://leepike.github.io/pubs/RTA-CPS.pdf}

\bibitem{cofer-taxi-20}
Cofer, D., Amundson, I., Sattigeri, R., Passi, A., Boggs, C., Smith, E.,
  Gilham, L., Byun, T., Rayadurgam, S.: Run-time assurance for learning-based
  aircraft taxiing. In: Digital Avionics Systems Conference (October 2020)

\bibitem{cofer-rta-22}
Cofer, D., Sattigeri, R., Amundson, I., Babar, J., Hasan, S., Smith, E.,
  Nukala, K., Osipychev, D., Moser, M., Paunicka, J., Margineantu, D.,
  Timmerman, L., Stringfield, J.: Flight test of a collision avoidance neural
  network with run-time assurance. In: Digital Avionics Systems Conference
  (September 2022)

\bibitem{ACC-xenofon}
Dai, S., Koutsoukos, X.: Safety analysis of integrated adaptive cruise control
  and lane keeping control using discrete-time models of port-hamiltonian
  systems. In: 2017 American Control Conference (ACC). pp. 2980--2985 (2017).
  \doi{10.23919/ACC.2017.7963404}

\bibitem{SOTER}
Desai, A., Ghosh, S., Seshia, S.A., Shankar, N., Tiwari, A.: Soter: A runtime
  assurance framework for programming safe robotics systems. In: 2019 49th
  Annual IEEE/IFIP International Conference on Dependable Systems and Networks
  (DSN). pp. 138--150. IEEE Computer Society, Los Alamitos, CA, USA (jun 2019).
  \doi{10.1109/DSN.2019.00027},
  \url{https://doi.ieeecomputersociety.org/10.1109/DSN.2019.00027}

\bibitem{SibaiOnlinePedest20}
Du, P., Huang, Z., Liu, T., Ji, T., Xu, K., Gao, Q., Sibai, H.,
  Driggs{-}Campbell, K.R., Mitra, S.: Online monitoring for safe
  pedestrian-vehicle interactions. In: 23rd {IEEE} International Conference on
  Intelligent Transportation Systems, {ITSC} 2020, Rhodes, Greece, September
  20-23, 2020. pp.~1--8. {IEEE} (2020). \doi{10.1109/ITSC45102.2020.9294366},
  \url{https://doi.org/10.1109/ITSC45102.2020.9294366}

\bibitem{dunlap2022run}
Dunlap, K.: Run Time Assurance for Intelligent Aerospace Control Systems. Ph.D.
  thesis, University of Cincinnati (2022)

\bibitem{dunlap2021comparing}
Dunlap, K., Hibbard, M., Mote, M., Hobbs, K.: Comparing run time assurance
  approaches for safe spacecraft docking. IEEE Control Systems Letters
  \textbf{6},  1849--1854 (2021)

\bibitem{DunlapHMH22}
Dunlap, K., Hibbard, M., Mote, M., Hobbs, K.: Comparing run time assurance
  approaches for safe spacecraft docking. {IEEE} Control. Syst. Lett.
  \textbf{6},  1849--1854 (2022). \doi{10.1109/LCSYS.2021.3135260},
  \url{https://doi.org/10.1109/LCSYS.2021.3135260}

\bibitem{eller2013test}
Eller, B., Stanfill, P., Turner, R., Whitcomb, S., Swihart, D., Burns, A.,
  Hobbs, K.: Test and evaluation of a modified f-16 analog flight control
  computer. In: AIAA Infotech@ Aerospace (I@ A) Conference. p.~4726 (2013)

\bibitem{fainekos2012verification}
Fainekos, G.E., Sankaranarayanan, S., Ueda, K., Yazarel, H.: Verification of
  automotive control applications using {S-TaLiRo}. In: American Control
  Conference (ACC), 2012. pp. 3567--3572. IEEE (2012)

\bibitem{fan2020fast}
Fan, C., Miller, K., Mitra, S.: Fast and guaranteed safe controller synthesis
  for nonlinear vehicle models. In: International Conference on Computer Aided
  Verification. pp. 629--652. Springer (2020)

\bibitem{FanQM18}
Fan, C., Qi, B., Mitra, S.: Data-driven formal reasoning and their applications
  in safety analysis of vehicle autonomy features. {IEEE} Design {\&} Test
  \textbf{35}(3),  31--38 (2018). \doi{10.1109/MDAT.2018.2799804},
  \url{https://doi.org/10.1109/MDAT.2018.2799804}

\bibitem{GarciaF15}
García, J., Fernández, F.: A comprehensive survey on safe reinforcement
  learning. J. Mach. Learn. Res.  \textbf{16},  1437--1480 (2015),
  \url{http://dblp.uni-trier.de/db/journals/jmlr/jmlr16.html#GarciaF15}

\bibitem{gurriet2018online}
Gurriet, T., Mote, M., Ames, A.D., Feron, E.: An online approach to active set
  invariance. In: 2018 IEEE Conference on Decision and Control (CDC). pp.
  3592--3599. IEEE (2018)

\bibitem{hahn2019omega}
Hahn, E.M., Perez, M., Schewe, S., Somenzi, F., Trivedi, A., Wojtczak, D.:
  Omega-regular objectives in model-free reinforcement learning. In: Tools and
  Algorithms for the Construction and Analysis of Systems: 25th International
  Conference, TACAS 2019, Held as Part of the European Joint Conferences on
  Theory and Practice of Software, ETAPS 2019, Prague, Czech Republic, April
  6--11, 2019, Proceedings, Part I. pp. 395--412. Springer (2019)

\bibitem{hasanbeig2019reinforcement}
Hasanbeig, M., Kantaros, Y., Abate, A., Kroening, D., Pappas, G.J., Lee, I.:
  Reinforcement learning for temporal logic control synthesis with
  probabilistic satisfaction guarantees. In: 2019 IEEE 58th conference on
  decision and control (CDC). pp. 5338--5343. IEEE (2019)

\bibitem{hibbard2022guaranteeing}
Hibbard, M., Topcu, U., Hobbs, K.: Guaranteeing safety via active-set
  invariance filters for multi-agent space systems with coupled dynamics. In:
  2022 American Control Conference (ACC). pp. 430--436. IEEE (2022)

\bibitem{hobbs2018space}
Hobbs, K., Heidlauf, P., Collins, A., Bak, S.: Space debris collision detection
  using reachability. In: ARCH@ ADHS. pp. 218--228 (2018)

\bibitem{Hobbs2021RunTA}
Hobbs, K., Mote, M.L., Abate, M., Coogan, S.D., Feron, E.: Run time assurance
  for safety-critical systems: An introduction to safety filtering approaches
  for complex control systems. ArXiv  \textbf{abs/2110.03506} (2021)

\bibitem{kanayama1990stable}
Kanayama, Y., Kimura, Y., Miyazaki, F., Noguchi, T.: A stable tracking control
  method for an autonomous mobile robot. In: Proceedings., IEEE International
  Conference on Robotics and Automation. pp. 384--389. IEEE (1990)

\bibitem{konighofer2020shield}
K{\"o}nighofer, B., Lorber, F., Jansen, N., Bloem, R.: Shield synthesis for
  reinforcement learning. In: Leveraging Applications of Formal Methods,
  Verification and Validation: Verification Principles: 9th International
  Symposium on Leveraging Applications of Formal Methods, ISoLA 2020, Rhodes,
  Greece, October 20--30, 2020, Proceedings, Part I 9. pp. 290--306. Springer
  (2020)

\bibitem{lavaei20}
Lavaei, A., Somenzi, F., Soudjani, S., Trivedi, A., Zamani, M.: Formal
  controller synthesis for continuous-space mdps via model-free reinforcement
  learning. In: Proceedings of the International Confenrence on Cyber-Physical
  Systems. pp. 98--107 (2020)

\bibitem{Verse23}
Li, Y., Zhu, H., Braught, K., Shen, K., Mitra, S.: Verse: {A} python library
  for reasoning about multi-agent hybrid system scenarios. CoRR
  \textbf{abs/2301.08714} (2023),
  \url{https://github.com/AutoVerse-ai/Verse-library}

\bibitem{LoosPN11}
Loos, S.M., Platzer, A., Nistor, L.: Adaptive cruise control: Hybrid,
  distributed, and now formally verified. In: {FM} 2011: Formal Methods - 17th
  International Symposium on Formal Methods, Limerick, Ireland, June 20-24,
  2011. Proceedings. pp. 42--56 (2011). \doi{10.1007/978-3-642-21437-0\_6},
  \url{https://doi.org/10.1007/978-3-642-21437-0\_6}

\bibitem{MehmoodSBSS22}
Mehmood, U., Sheikhi, S., Bak, S., Smolka, S.A., Stoller, S.D.: The black-box
  simplex architecture for runtime assurance of autonomous {CPS}. In: {NASA}
  Formal Methods - 14th International Symposium, {NFM} 2022, Pasadena, CA, USA,
  May 24-27, 2022, Proceedings. pp. 231--250 (2022).
  \doi{10.1007/978-3-031-06773-0\_12},
  \url{https://doi.org/10.1007/978-3-031-06773-0\_12}

\bibitem{ttc}
Minderhoud, M.M., Bovy, P.H.: Extended time-to-collision measures for road
  traffic safety assessment. Accident Analysis \& Prevention  \textbf{33}(1),
  89--97 (2001)

\bibitem{mote2021natural}
Mote, M.L., Hays, C.W., Collins, A., Feron, E., Hobbs, K.L.: Natural
  motion-based trajectories for automatic spacecraft collision avoidance during
  proximity operations. In: 2021 IEEE Aerospace Conference (50100). pp. 1--12.
  IEEE (2021)

\bibitem{musau2022icaa}
Musau, P., Hamilton, N., Lopez, D.M., Robinette, P., Johnson, T.T.: On using
  real-time reachability for the safety assurance of machine learning
  controllers. In: 2022 IEEE International Conference on Assured Autonomy
  (ICAA). pp. 1--10 (2022). \doi{10.1109/ICAA52185.2022.00010}

\bibitem{Puterman:1994}
Puterman, M.L.: Markov Decision Processes: Discrete Stochastic Dynamic
  Programming. John Wiley \& Sons, New York (1994)

\bibitem{ramadge1987supervisory}
Ramadge, P.J., Wonham, W.M.: Supervisory control of a class of discrete event
  processes. SIAM journal on control and optimization  \textbf{25}(1),
  206--230 (1987)

\bibitem{RTA-JAX-abs-2209-01120}
Ravaioli, U., Dunlap, K., Hobbs, K.: A universal framework for generalized run
  time assurance with {JAX} automatic differentiation. CoRR
  \textbf{abs/2209.01120} (2022). \doi{10.48550/arXiv.2209.01120},
  \url{https://doi.org/10.48550/arXiv.2209.01120}

\bibitem{Hobbs-SafeRL22}
Ravaioli, U.J., Cunningham, J., McCarroll, J., Gangal, V., Dunlap, K., Hobbs,
  K.L.: Safe reinforcement learning benchmark environments for aerospace
  control systems. In: 2022 IEEE Aerospace Conference (AERO). pp. 1--20 (2022).
  \doi{10.1109/AERO53065.2022.9843750}

\bibitem{ribichini2003efficient}
Ribichini, G., Frazzoli, E.: Efficient coordination of multiple-aircraft
  systems. In: 42nd IEEE International Conference on Decision and Control (IEEE
  Cat. No. 03CH37475). vol.~1, pp. 1035--1040. IEEE (2003)

\bibitem{sadigh2014learning}
Sadigh, D., Kim, E.S., Coogan, S., Sastry, S.S., Seshia, S.A.: A learning based
  approach to control synthesis of markov decision processes for linear
  temporal logic specifications. In: 53rd IEEE Conference on Decision and
  Control. pp. 1091--1096. IEEE (2014)

\bibitem{RTA-AAS2020}
Schierman, J., DeVore, M., Richards, N., Clark, M.: Runtime assurance for
  autonomous aerospace systems. In: Journal of Guidance, Control, and Dynamics,
  Vol. 43, No. 12 (2020)

\bibitem{RTA-VV-SC-FCS-2008}
Schierman, J., Ward, D., Dutoi, B., et~al.: Run-time verification and
  validation for safety-critical flight control systems. In: AIAA Paper 2008-
  6338, Proceedings of the AIAA Guidance, Navigation, and Control Conference,
  Honolulu, Hawaii, Aug., 2008

\bibitem{PPO}
Schulman, J., Wolski, F., Dhariwal, P., Radford, A., Klimov, O.: Proximal
  policy optimization algorithms (2017). \doi{10.48550/ARXIV.1707.06347},
  \url{https://arxiv.org/abs/1707.06347}

\bibitem{sha2001using}
Sha, L., et~al.: Using simplicity to control complexity. IEEE Software
  \textbf{18}(4),  20--28 (2001)

\bibitem{wadley2013development}
Wadley, J., Jones, S., Stoner, D., Griffin, E., Swihart, D., Hobbs, K., Burns,
  A., Bier, J.: Development of an automatic aircraft collision avoidance system
  for fighter aircraft. In: AIAA Infotech@ Aerospace (I@ A) Conference. p.~4727
  (2013)

\end{thebibliography}

%\appendix
%\input{appendix}
\end{document}